\documentclass[11pt]{amsart} 
\usepackage{amsmath}
\usepackage{amsfonts}
\usepackage{amssymb}
\usepackage{graphicx}
\usepackage{amsthm}
\usepackage{setspace}
\usepackage{multirow}
\usepackage{longtable}

\numberwithin{equation}{section}

\theoremstyle{plain}
\newtheorem{theorem}{Theorem}[section]

\newtheorem{remark}[theorem]{Remark}

\newcommand{\E}{\mathbb{E}}

\title{The microstructure of high frequency markets}
\date{August 24, 2014}
\author{Ren\'e Carmona and Kevin Webster}
\address{ORFE, Bendheim Center for Finance\\
Princeton University}

\begin{document}

\maketitle

\begin{abstract}
We present a novel approach to describing the microstructure of high frequency trading using two key elements. First we introduce a new notion of informed trader which we starkly contrast to current \emph{informed trader} models. We describe the exact nature of the `superior information' high frequency traders have access to, and how these agents differ from the more standard `insider traders' described in past papers. This then leads to a model and an empirical analysis of the data which strongly supports our claims. The second key element is a rigorous description of clearing conditions on a limit order book and how to derive correct formulas for such a market. From a theoretical point of view, this allows the exact identification of two frictions in the market, one of which is intimately linked to our notion of `superior information'. Empirically, we show that ignoring these frictions can misrepresent the wealth exchanged on the market by $50\%$. Finally, we showcase two applications of our approach: we measure the profits made by high frequency traders on NASDAQ and re-visit the standard Black - Scholes model to determine how trading frictions alter the delta-hedging strategy.
\end{abstract}

\section{Introduction}
It is impossible to begin a paper on informed trading and the microstructure of markets without citing the fundamental work of Kyle. This is very fortunate, as Kyle's model is a perfect example of an elegant yet sophisticated model that describes with great accuracy an important phenomenon in financial markets. It is also rather surprising, as the paper was written nearly three decades ago, when the agents and microstructure of markets was very different from that of modern days.

Our paper is a tribute to Kyle's model. Rather than apply the original findings of his ground-breaking paper in circumstances that it did not seek to address, we actively engage with it to reflect the advent of high frequency trading. This allows us to identify the main common points and differences between the world of 1985 and 2014. The common points are adverse selection, the impact of informed trades on prices and the ensuing information rent. The differences are the \emph{nature} of the informed agents, their trading strategies and the impact they have on the financial system.

At the heart of the discussion are two key notions: information and trading. The combination of the two leads to adverse selection. The former requires the reader to understand the nature of the advantage the informed trader enjoys. This can be done through a direct microscopic analysis of agents: in this case high frequency traders. The latter point is more technical in nature. Before any economic conclusion can be reached one has to delve into the exact rules of electronic exchanges. This allows the identification of the exact sources of profits and losses for a trader trading electronically. Both can be tackled on a theoretical level, but also almost physically grasped through the massive amount of data present in modern markets.

Besides the original work by Kyle (\cite{Kyle}), other microstructure papers focus on inventory managing market makers (\cite{Amihud, Garman, Stoll, O'Hara5}). These models work under a different market microstructure than the one we study, and are designed to provide qualitative insights on the price formation. Our focus is more on high frequency traders, their profits and the empirically measurable impact they have on the market. 

We posit that models of high frequency markets should include our clearing equation and adverse selection model. The main advantage of such an approach is that it does not require a model for the microstructure noise: wealth is perfectly matched to data using only market-observable quantities. We do not track the true price, only the observed mid-price. We refer to \cite{Yacine_noise} and the reference therein for the very rich literature on microstructure noise and the corresponding models of high frequency markets. Both the paper by Y.~Ait-Sahalia, P.~A. Mykland and L.~Zhang \cite{Yacine_noise} and the book by Y.~Ait-Sahalia and J.~Jacod \cite{Yacine} have been instrumental in the empirical section of our paper, as they highlight the technicalities surrounding econometric tests in high frequency markets. Unlike these references however, our paper aims to measure and match market quantities without resorting to the notion of microstructure noise.

An alternative way to model wealth in a high frequency world is the Almgren and Chriss approach \cite{Almgren}, which relies on a linear price-impact relationship. This approach to adverse selection is more explicit, though more rigid than ours. It is a purely descriptive, agent-less model. Instead of being driven by data, its main focus is tractability for a specific type of applications: optimal execution. Because of this, calibration is not as straightforward in the Almgren and Chriss model as in ours. However, the linear structure leads to closed-form formulas for a larger class of problems, mostly found in the optimal execution literature (see \cite{Alfonsi, Wang}).

From an empirical point of view, we study a market for a financial asset traded electronically via a limit order book. Market orders are assumed to arrive at high frequency and all trades execute either at the best bid or ask price. Moreover, hidden orders are discarded. The appendix contains more information on our data set from NASDAQ. 

A recent pair of papers using high frequency data is \cite{O'Hara3, O'Hara2}. Both papers by O'Hara et al. estimate adverse selection by measuring the imbalance between buy and sell orders. One of the main technical tools advocated by these authors is the use of an \emph{event-based clock} to model electronic markets, in line with the practice of high frequency trading. We adopt this event-based point of view, but choose instead to measure adverse selection by the relationship between executed limit orders and price changes.

We illustrate the significant impact of the trading frictions we identify with an application to option hedging by replication in the spirit the seminal work of F. Black and M. Scholes. In our framework, European options are priced with the same formula, but the volatility needs to be changed to incorporate the bid-ask spread. Perfect replication is still possible, albeit with a twist: negative gamma options can only be replicated with limit orders while positive gamma options are hedged with market orders. 

These results on option pricing are reminiscent of, but distinct from,  Leland's strategy (see \cite{Leland} for example), in which pricing of European options under transaction costs is explored. The major differences are 1) the sharp distinction between the use of limit and market orders as hedging tools;  2) the impact of adverse selection which is absent in Leland's model. At a procedural level, Leland adds transaction costs to a model already set up in continuous time, while in our approach, we model the market microstructure frictions at the microscopic level, before the extension to a macroscopic, continuous time model. 

The rest of the paper is organized as follows. Section \ref{sec:superior_information} describes our notion of \emph{superior information} to describe high frequency traders. It is contrasted to the usual notion of insider information. A simple model is presented and tested on data. Section \ref{sec:clearing_condition} begins with a technical description of the clearing mechanism on a high frequency market. This leads to novel equations that decompose a trader's wealth into frictionless wealth, transaction costs and adverse selection costs. An empirical analysis shows the importance of the so far neglected terms. Finally, section \ref{sec:applications} presents two applications of our framework. The first measures the profits made by high frequency traders, the effective spread captured by liquidity providers and provides an econometric test for adverse selection. The paper concludes with an application to option hedging in the spirit of the standard Black - Scholes replication argument.

\section{Superior Information of High Frequency Traders}\label{sec:superior_information} 
This section describes high frequency traders as \emph{informed agents}. The main message is that the nature of their information, and the risk they carry, is \emph{radically different} from what has been proposed in the microstructure literature to date. Our canonical reference to contrast the two information structures is the fundamental work of Kyle on insider trading.
  
\subsection{Insider information or superior information?}
An informed trader is an agent who holds more accurate information on an asset than the rest of the market. The most extreme and well-known example is \emph{insider trading}, where a trader knows ahead of time the true future value of an asset. To be precise, let us describe the exact nature of the information an insider has access to, how he trades to take advantage of his information and what market risks this strategy entails.

The insider trader knows the true value of the asset at some future time. This tells him the \emph{fundamental} price of the asset and a \emph{time frame} on which the market value has to converge to this price level. The insider trader has no additional information on the current supply and demand of the asset and hence will not be able to predict short-term price movements. As expertly described in \cite{Kyle}, if the current market price is lower than the true value of the asset, the trader will accumulate a position over a long period of time, regardless of the price movements and make a profit. If the market price is higher than the true value, the insider will sell the asset over an extended period of time. His edge is the knowledge of the fundamental price of the asset. The main risk the insider runs is marking his inventory to the market. The exact flow of supply and demand can erode his profits through multiple channels. First, if he trades too fast, the inside information will be incorporated faster into the price, leading to less profits. Furthermore, short-term price fluctuations due to supply and demand dynamics can make his marked-to-market profits volatile. He may be asked to satisfy margin requirements by the exchange or regulatory collateral requirements. These will make his return on investment less attractive.

Now consider a high frequency trader. Unlike the insider trader, she has \emph{no additional information} on the fundamental value of an asset. She does \emph{not} know the true price of the asset at some future point in time. Instead, a high frequency trader has a much more accurate picture of the \emph{current supply and demand} of the asset. This advantage can manifest itself in a number of ways: she may have access to more market information\footnote{Without being exhaustive, examples of such an information bias are:
\begin{enumerate}
\item The ability to see the full order book beyond the best bid or ask price.
\item Access to other pools of liquidity, such as alternative exchanges or dark pools.
\item More processing and storing power to pour over the massive data feeds. 
\end{enumerate}} or she may simply process the information faster\footnote{The canonical and most extreme example of that is \emph{latency arbitrage}. There have also been cases of alleged front-running.}. But just as in the literature on insider trading, the underlying cause for the information advantage is not as important as the structure of the information itself -and how it is used-. In the case of high frequency trading, the information is such that the trader can \emph{predict with near perfect accuracy} the next \emph{price movement}. She profits from this knowledge by trading in anticipation of these short-term price movements, buying before the price moves up and selling before the price moves down. She therefore makes a profit on each trade marked-to-market, but runs the risk of accumulating an unwanted inventory which may turn against her in the long run.

So, in a nutshell, the insider trader and the high-frequency trader can be seen as mirror images of each other:
\begin{itemize}
\item The insider predicts a price level, which corresponds to some fundamental information on the asset. 
\item The high frequency trader predicts price \emph{fluctuations} due to information on the market's state of supply and demand for the asset.
\item The insider makes money accumulating a position in the long run. All trades go in the same direction.
\item The high frequency trader makes marked-to-market profits trade by trade. Trades oscillate rapidly between buys and sells.
\item The insider's risk is the marked-to-market value of his portfolio.
\item The high frequency trader's risk is to accumulate an inventory. 
\end{itemize}

The current literature does not distinguish between these two types of \emph{informed} traders. We will refer to the first one as traders with \emph{inside} information. The second type of trader will be said to have \emph{superior} information. After describing these traders in words, we propose a simple model distinguishing their behavior.

\subsection{A toy model}
Let $(p_n)_{n=1...N}$ be a sequence of prices and $(L_n)_{n=1...N}$ a trader's inventory \emph{before} the trade at time $n$, which will happen at the market price $p_n$ (up to some bid-ask spread). 

The standard model for a trader with inside information is to say that the trader has access to the information $p_N$ in addition to the current price $p_n$, but not the intermediate prices. The trading volume at time $n$ is $\Delta_{n} L = L_{n+1}- L_n$, the inventory after the trade minus the inventory before the trade. For an insider, it will be of the same sign as $p_N - p_n$ to guarantee profits at time $N$. One can find an example of a more explicit trading strategy for insiders in \cite{Kyle}.

Our proposed equivalent for traders with superior information is to give the trader access to $p_{n+1}$ instead of $p_N$. The trader with superior information therefore knows $p_n$ as well as the forward-looking price increment $\Delta_{n} p = p_{n+1} - p_n$, but not the prices in the further future. Her trade volume $\Delta_n L$ will be of the same sign as $\Delta_{n} p$ thanks to this \emph{look-ahead} option.

For an uninformed trader, $\Delta_n L$ is noise and therefore uncorrelated to both $p_N - p_n$ and $\Delta_n p$. 

One could build a full equilibrium model to get a more explicit relationship between $\Delta_n L$ and $\Delta_n p$. The bare bones model $\Delta_n p \Delta_n L >0$ for traders with superior information, $\E \Delta_n p \Delta_n L = 0$ for noise traders will suffice for this paper. Note that this relationship can also be interpreted as a form of \emph{price impact}, as trading in advance to a price increment cannot be empirically distinguished from the trade moving the price. Indeed, some high frequency trading strategies rely on a thin order book to \emph{cause} short-term price movements instead of predicting them.

\subsection{Empirical analysis}
A straightforward way to quantify the above relationship for a given data set is to compute the proportion of trades such that
\begin{enumerate}
\item $\Delta_n L \Delta_n p >0$; \label{price_impact}
\item $\Delta_n L \Delta_n p =0$; \label{no_price_impact}
\item $\Delta_n L \Delta_n p <0$. \label{inverse_price_impact}
\end{enumerate} 
Case (\ref{price_impact}) means that the trade led to price impact. Case (\ref{no_price_impact}) means that the trade left the midprice unchanged and was caused by 'noise' trading. Case (\ref{inverse_price_impact}) corresponds to reverse price impact and is also due to noise trading. 

Not all trades with price impact come from informed traders: they could also be the product of lucky noise trades. However, both trades without price impact and inverse price impact must come from noise traders. Because for these traders, $\E\Delta_n L \Delta_n p = 0$, the amount of false positives when equating price impact with informed trading is equal to the number of inverse price impact trades: the lower the amount of trades with inverse price impact, the higher the proportion of trades with price impact that are due to traders with superior information rather than noise traders.

\begin{center}
\tiny{
		\begin{longtable}{|l|c|c|c|c|}
			stock symbol & total nb of trades & with price impact & without price impact & reverse price impact \\
				\hline
					 AA       & 5115 & 325 (6.35\%) & 4789 & 1 (0.02\%) \\ 
					 AAPL     & 46709 & 16896 (36.17\%) & 27466 & 2347 (5.03\%) \\ 
					 ADBE     & 9031 & 2270 (25.13\%) & 6613 & 148 (1.64\%) \\ 
					 AGN      & 1884 & 1073 (56.95\%) & 656 & 155 (8.23\%) \\ 
					 AINV     & 1882 & 199 (10.57\%) & 1678 & 5 (0.27\%) \\ 
					 AMAT     & 9383 & 524 (5.58\%) & 8854 & 5 (0.06\%) \\ 
					 AMED     & 697 & 195 (27.97\%) & 478 & 24 (3.45\%) \\ 
					 AMGN     & 15200 & 6494 (42.72\%) & 8161 & 545 (3.59\%) \\ 
					 AMZN     & 11966 & 6098 (50.96\%) & 5219 & 649 (5.43\%) \\ 
					 ANGO     & 482 & 128 (26.55\%) & 337 & 17 (3.53\%) \\ 
					 APOG     & 968 & 411 (42.45\%) & 501 & 56 (5.79\%) \\ 
					 ARCC     & 2242 & 246 (10.97\%) & 1992 & 4 (0.18\%) \\ 
					 AXP      & 14395 & 4070 (28.27\%) & 9982 & 343 (2.39\%) \\ 
					 AYI      & 732 & 438 (59.83\%) & 231 & 63 (8.61\%) \\ 
					 AZZ      & 102 & 63 (61.76\%) & 21 & 18 (17.65\%) \\ 
					 BAS      & 469 & 268 (57.14\%) & 163 & 38 (8.11\%) \\ 
					 BHI      & 9648 & 3474 (36\%) & 5818 & 356 (3.69\%) \\ 
					 BIIB     & 4270 & 2287 (53.55\%) & 1658 & 325 (7.62\%) \\ 
					 BRCM     & 18273 & 3310 (18.11\%) & 14839 & 124 (0.68\%) \\ 
					 BRE      & 1141 & 490 (42.94\%) & 580 & 71 (6.23\%) \\ 
					 BXS      & 672 & 272 (40.47\%) & 377 & 23 (3.43\%) \\ 
					 BZ       & 299 & 110 (36.78\%) & 167 & 22 (7.36\%) \\ 
					 CB       & 1906 & 924 (48.47\%) & 879 & 103 (5.41\%) \\ 
					 CBEY     & 519 & 168 (32.36\%) & 315 & 36 (6.94\%) \\ 
					 CBT      & 645 & 334 (51.78\%) & 242 & 69 (10.7\%) \\ 
					 CBZ      & 121 & 47 (38.84\%) & 66 & 8 (6.62\%) \\ 
					 CDR      & 166 & 58 (34.93\%) & 104 & 4 (2.41\%) \\ 
					 CELG     & 7699 & 3828 (49.72\%) & 3339 & 532 (6.91\%) \\ 
					 CETV     & 248 & 72 (29.03\%) & 166 & 10 (4.04\%) \\ 
					 CKH      & 212 & 123 (58.01\%) & 63 & 26 (12.27\%) \\ 
					 CMCSA    & 17870 & 2410 (13.48\%) & 15424 & 36 (0.21\%) \\ 
					 CNQR     & 976 & 467 (47.84\%) & 443 & 66 (6.77\%) \\ 
					 COO      & 221 & 127 (57.46\%) & 68 & 26 (11.77\%) \\ 
					 COST     & 8229 & 3578 (43.48\%) & 4346 & 305 (3.71\%) \\ 
					 CPSI     & 334 & 167 (50\%) & 136 & 31 (9.29\%) \\ 
					 CPWR     & 1855 & 273 (14.71\%) & 1576 & 6 (0.33\%) \\ 
					 CR       & 590 & 302 (51.18\%) & 238 & 50 (8.48\%) \\ 
					 CRI      & 925 & 487 (52.64\%) & 369 & 69 (7.46\%) \\ 
					 CSCO     & 20155 & 1315 (6.52\%) & 18836 & 4 (0.02\%) \\ 
					 CSE      & 1213 & 169 (13.93\%) & 1041 & 3 (0.25\%) \\ 
					 CSL      & 377 & 189 (50.13\%) & 145 & 43 (11.41\%) \\ 
					 CTRN     & 141 & 62 (43.97\%) & 68 & 11 (7.81\%) \\
					 CTSH     & 7991 & 2958 (37.01\%) & 4821 & 212 (2.66\%) \\ 
					 DCOM     & 312 & 132 (42.3\%) & 148 & 32 (10.26\%) \\ 
					 DELL     & 3741 & 198 (5.29\%) & 3542 & 1 (0.03\%) \\ 
					 DIS      & 12956 & 2924 (22.56\%) & 9895 & 137 (1.06\%) \\ 
					 DK       & 1045 & 549 (52.53\%) & 381 & 115 (11.01\%) \\ 
					 DOW      & 8823 & 1682 (19.06\%) & 7075 & 66 (0.75\%) \\ 
					 EBAY     & 47060 & 11716 (24.89\%) & 34571 & 773 (1.65\%) \\ 
					 ESRX     & 10745 & 2985 (27.78\%) & 7543 & 217 (2.02\%) \\ 
					 EWBC     & 6391 & 1335 (20.88\%) & 4983 & 73 (1.15\%) \\ 
					 FCN      & 248 & 135 (54.43\%) & 98 & 15 (6.05\%) \\ 
					 FFIC     & 200 & 79 (39.5\%) & 104 & 17 (8.51\%) \\ 
					 FL       & 3877 & 1103 (28.44\%) & 2716 & 58 (1.5\%) \\ 
					 FMER     & 3747 & 751 (20.04\%) & 2965 & 31 (0.83\%) \\ 
					 FPO      & 256 & 146 (57.03\%) & 97 & 13 (5.08\%) \\ 
					 FRED     & 751 & 283 (37.68\%) & 435 & 33 (4.4\%) \\ 
					 FULT     & 3202 & 455 (14.2\%) & 2736 & 11 (0.35\%) \\ 
					 GAS      & 643 & 306 (47.58\%) & 277 & 60 (9.34\%) \\ 
					 GE       & 12968 & 916 (7.06\%) & 12048 & 4 (0.04\%) \\ 
					 GILD     & 20404 & 4745 (23.25\%) & 15370 & 289 (1.42\%) \\ 
					 GLW      & 7346 & 403 (5.48\%) & 6942 & 1 (0.02\%) \\ 
					 GOOG     & 8594 & 3919 (45.6\%) & 4066 & 609 (7.09\%) \\ 
					 GPS      & 7377 & 1815 (24.6\%) & 5465 & 97 (1.32\%) \\ 
					 HON      & 5813 & 2170 (37.33\%) & 3460 & 183 (3.15\%) \\ 
					 HPQ      & 14918 & 1823 (12.22\%) & 13076 & 19 (0.13\%) \\ 
					 IMGN     & 2416 & 708 (29.3\%) & 1627 & 81 (3.36\%) \\ 
					 INTC     & 28221 & 1192 (4.22\%) & 27022 & 7 (0.03\%) \\ 
					 IPAR     & 172 & 66 (38.37\%) & 83 & 23 (13.38\%) \\ 
					 ISIL     & 2329 & 382 (16.4\%) & 1938 & 9 (0.39\%) \\ 
					 ISRG     & 2415 & 1189 (49.23\%) & 1026 & 200 (8.29\%) \\ 
					 JKHY     & 1002 & 428 (42.71\%) & 515 & 59 (5.89\%) \\ 
					 KMB      & 2608 & 1414 (54.21\%) & 981 & 213 (8.17\%) \\ 
					 KR       & 6546 & 1253 (19.14\%) & 5263 & 30 (0.46\%) \\ 
					 LANC     & 226 & 100 (44.24\%) & 105 & 21 (9.3\%) \\ 
					 LECO     & 1093 & 526 (48.12\%) & 503 & 64 (5.86\%) \\ 
					 LPNT     & 1876 & 832 (44.34\%) & 931 & 113 (6.03\%) \\ 
					 LSTR     & 1156 & 574 (49.65\%) & 493 & 89 (7.7\%) \\ 
					 MAKO     & 1286 & 339 (26.36\%) & 920 & 27 (2.1\%) \\ 
					 MANT     & 482 & 201 (41.7\%) & 242 & 39 (8.1\%) \\ 
					 MDCO     & 1737 & 753 (43.35\%) & 900 & 84 (4.84\%) \\ 
					 MELI     & 812 & 361 (44.45\%) & 380 & 71 (8.75\%)  \\
				 MIG      & 528 & 71 (13.44\%) & 451 & 6 (1.14\%) \\ 
				 MMM      & 5036 & 2074 (41.18\%) & 2682 & 280 (5.56\%) \\ 
				 MOD      & 110 & 47 (42.72\%) & 56 & 7 (6.37\%) \\ 
				 MOS      & 3754 & 1549 (41.26\%) & 2016 & 189 (5.04\%) \\ 
				 MRTN     & 378 & 141 (37.3\%) & 212 & 25 (6.62\%) \\ 
				 MXWL     & 496 & 175 (35.28\%) & 313 & 8 (1.62\%) \\ 
				 NSR      & 556 & 258 (46.4\%) & 247 & 51 (9.18\%) \\ 
				 NUS      & 1300 & 629 (48.38\%) & 561 & 110 (8.47\%) \\ 
				 NXTM     & 523 & 140 (26.76\%) & 359 & 24 (4.59\%) \\ 
				 PBH      & 123 & 60 (48.78\%) & 43 & 20 (16.27\%) \\ 
				 PFE      & 14626 & 1314 (8.98\%) & 13303 & 9 (0.07\%) \\ 
				 PG       & 18615 & 3768 (20.24\%) & 14658 & 189 (1.02\%) \\ 
				 PNC      & 5277 & 2015 (38.18\%) & 3075 & 187 (3.55\%) \\ 
				 PNY      & 540 & 248 (45.92\%) & 231 & 61 (11.3\%) \\ 
				 PTP      & 513 & 294 (57.3\%) & 172 & 47 (9.17\%) \\ 
				 RIGL     & 2079 & 371 (17.84\%) & 1693 & 15 (0.73\%) \\ 
				 ROC      & 968 & 523 (54.02\%) & 336 & 109 (11.27\%) \\ 
				 ROCK     & 743 & 214 (28.8\%) & 507 & 22 (2.97\%) \\ 
				 SF       & 440 & 242 (55\%) & 148 & 50 (11.37\%) \\ 
				 SFG      & 170 & 98 (57.64\%) & 48 & 24 (14.12\%) \\ 
				 SWN      & 13815 & 3635 (26.31\%) & 9953 & 227 (1.65\%)
		\end{longtable}

{Table 1. Proportion of trades with price impact, without price impact and with reverse price impact on  04/18/13 for 103 stocks.}
}
\end{center}

An important empirical remark is that the number of trades exhibiting inverse price impact is very small. This is probably due to the discrete nature of prices and the bid-ask spread on high frequency markets\footnote{The so called 'rounding' effect.}. This means that nearly all the trades with price impact were due to traders with superior information, and that noise traders typically have no price impact. The empirical reality is therefore even simpler than our model!

\section{The Clearing Condition on a Limit Order Book}\label{sec:clearing_condition}
On a technical level, the aim of this section is to describe accurately the evolution of the wealth of a trader trading on a high frequency order book. This is the equivalent of a \emph{clearing condition} and depends crucially on whether the trader is trading via limit orders or market orders. The former will be qualified as a \emph{passive} trader, while we describe the latter as an \emph{active} trader. Once these technical equations have been derived, we can extract from a trading flow the amount of transaction cost paid as well as the money lost due to adverse selection by traders with superior information. As a useful practical consequence, we can also perfectly track trading profits and avoid the notion of microstructure noise.

\subsection{Setup}
We first define the quantities of interest. Let $p_n$ be the midprice \emph{before} the trade happens at time $n$. The midprice is defined as the midpoint between the bid and the ask and is often abbreviated to 'the mid'. Similarly, let $s_n$ be the bid-ask spread defined as the difference between the best ask and the best bid.  Note that the number $s_n /2$, the distance between the mid and either leg of the bid-ask, is often called 'the spread'. 

We now single out a trader who's profits and losses we wish to follow. Denote by $L_n$ his inventory, that is, his net position in the traded asset. Finally, define by $K_n$ the amount of cash the trader is holding and assume that his position is self-financing in the sense that the changes in the inventory and cash account can only come from trading on the high frequency order book.

The quantities $p_n$, $s_n$, $L_n$ and $K_n$ are necessary and sufficient to, from an accountant's perspective, summarize the trader's strategy given our set of hypotheses. They are therefore the primary data of our model. It is common however to summarize the trader's profits and losses by a single number called wealth. 

The most basic way to define wealth from the primary data is to say it is the \emph{marked to the mid value of the inventory plus the cash holdings}. If we denote by $X_n$ the trader's wealth before trade $n$, this leads to the equation:
\begin{equation}
X_n = L_n p_n + K_n 
\end{equation}
This is the main quantity of interest we would like to track, both in the theoretical model, and through our empirical analyzes of trading data.

\subsection{Clearing a trade}
We assume the position of the trader to be self-financing. This imposes relationships between $p_n$, $s_n$, $L_n$ and $K_n$ based on the clearing rules that underpin trades on a limit order book. There are five cases to distinguish, depending on whether the trader
\begin{enumerate}
\item triggers a buy with a market order; \label{buy_MO}
\item triggers a sell with a market order; \label{sell_MO}
\item has his buy limit order executed; \label{buy_LO}
\item has his sell limit order executed; \label{sell_LO}
\item is not part of the current trade. \label{no_trade}
\end{enumerate}
In case \ref{buy_MO}, the trader buys at the ask and therefore
\begin{equation}
K_{n+1} - K_n = - \left(p_n +s_n/2\right)\left(L_{n+1} - L_n\right)
\end{equation}
with $L_{n+1} - L_n >0$.

In case \ref{sell_MO}, the trader sells at the bid and therefore
\begin{equation}
K_{n+1} - K_n = - \left(p_n  -s_n/2\right)\left(L_{n+1} - L_n\right)
\end{equation}
with $L_{n+1} - L_n <0$.

In case \ref{buy_LO}, the trader buys at the bid and therefore
\begin{equation}
K_{n+1} - K_n = - \left(p_n -s_n/2\right)\left(L_{n+1} - L_n\right)
\end{equation}
with $L_{n+1} - L_n >0$.

In case \ref{sell_LO}, the trader sells at the ask and therefore
\begin{equation}
K_{n+1} - K_n = - \left(p_n +s_n/2\right)\left(L_{n+1} - L_n\right)
\end{equation}
with $L_{n+1} - L_n<0$.

Finally, in case \ref{no_trade}, $K_{n+1}= K_{n}$ and $L_{n+1} = L_n$.

\vskip 6pt\noindent
These five cases can be summarized by the equation:
\begin{equation}
\Delta_n K = - p_n \Delta_n L \pm \frac{s_n}{2} \left|\Delta_n L\right| \label{self_financing_K}
\end{equation}
where $\Delta_n K$ stands for $K_{n+1} - K_{n}$ and ''$\pm$'' is defined as ''$+$'' when trading with limit orders and ''$-$'' when trading with market orders. This is our clearing condition. It differs from that used in the Black and Scholes model, where $\Delta_n K = - p_{n+1} \Delta_n L$, unless the next mid price is equal to the previous bid or ask price.

This can be related to the wealth of a self-financing trading strategy by plugging the equation into the definition $X_n = L_n p_n + K_n$. This leads to the equation
\begin{equation}
\Delta_n X = L_n \Delta_n p \pm \frac{s_n}{2} \left|\Delta_n L\right| + \Delta_n L \Delta_n p \label{self_financing_discrete}
\end{equation}
where $\pm$ has the meaning defined above. These are the trading profits of a trader.

\subsection{Decomposing profits and losses}

Equation \eqref{self_financing_discrete} is a simple accounting rule for updating the wealth of a portfolio when trading on a high frequency order book. 
It could be stated in several different equivalent forms. We chose to write it in this way to highlight the last two terms. These differ from the commonly accepted rule in the frictionless case. We comment on the three terms separately:
\subsubsection{Frictionless wealth}
The standard self-financing equation used in a frictionless market model is
\begin{equation}
\Delta_n X = L_n \Delta_n p
\end{equation}
which coincides with our equation \eqref{self_financing_discrete} when the trader holds his position without trading. The rational for calling this the \emph{frictionless} wealth equation is the notion that whether the trader is holding the position or re-balancing it, should not affect his wealth. Whether this holds true in practice depends upon how close the values of the two other terms of equation \eqref{self_financing_discrete} are from canceling each other.

\subsubsection{Transaction costs}
The term
\begin{equation}
\pm \frac{s_n}{2} \left|\Delta_n L\right|
\end{equation}
has a clear interpretation as the spread captured or paid during the transaction. If one trades with limit orders, the spread is captured for each unit of traded asset while for market orders, the spread is paid. This term, all others being equal, favors trading with limit orders and constitutes a first friction on high frequency markets. Note that the exact form of this transaction cost term hinges on the assumption that all trades are made at the best bid or ask price. In the general case, transaction costs are given by a convex function of the trade volume. We leave the theoretical discussion of this fact to the appendix.

\subsubsection{Price Impact from adverse selection}
The third term is non-standard, and to the best of our knowledge, such a price impact term has never been identified and used in this direct form as a contribution to wealth change:
\begin{equation}
\Delta_n p\Delta_n L.
\end{equation}
We now propose an interpretation as a form of adverse selection. This links our clearing condition to the previous discussion on traders  with superior information. First, note that while the passive trader 'controls' the transaction cost term (by choosing how far to placer her limit orders from the mid price), the active trader controls this term through the trade volume $\Delta_n L$. If the active trader has no superior information, then $\Delta_n L$ will be uncorrelated to  $\Delta_n p$ and this term averages out to zero. If the active trader has superior information, then $\Delta_n p \Delta_n L >0$ from the active trader's perspective, and $\Delta_n p \Delta_n L <0$ from the passive trader's perspective. The trade will empirically exhibit price impact and the passive trader will lose money to the active trader.

\subsection{Empirical analysis}
We now compare empirically three clearing equations with the empirical exchange of wealth observed on the market.

The standard way to model wealth in the academic finance literature is to use the \emph{frictionless clearing equation}:
\begin{equation}
\Delta_n X = L_n \Delta_n p, \label{wealth_frictionless} 
\end{equation}
where $p$ is not the midprice, while the 'fair' price of the asset, which in an efficient market, is assumed to be a martingale. This equation does not match the data when we use the midprice and the concept of 'microstructure noise' can be invoked to explain this significant gap. In a frictionless market, microstructure does not matter and profits come solely from longer term views on the market. 

Transaction costs can be added, leading to the \emph{clearing equation with transaction costs}:
\begin{equation}
\Delta_n X = L_n \Delta_n p \pm \frac{s_n}{2} \left|\Delta_n L\right|\label{wealth_transaction_costs} 
\end{equation}
While this equation takes into account the spread, it ignores price impact. Is is however a good model for agents who are using market orders and who do not adversely select the market. This can be the case for low frequency traders who do not optimize their execution and only use long-term views to trade.

Finally, incorporating the price impact term leads to a complete picture of wealth on a high frequency market:
\begin{equation}
\Delta_n X = L_n \Delta_n p \pm \frac{s_n}{2} \left|\Delta_n L\right| + \Delta_n L \Delta_n p. \label{wealth_complete}
\end{equation}
This description forgoes the notion of microstructure noise, as wealth can be perfectly tracked using directly measurable market quantities.

\begin{figure}[htbp]
	\centering
		\includegraphics[width=0.9\textwidth]{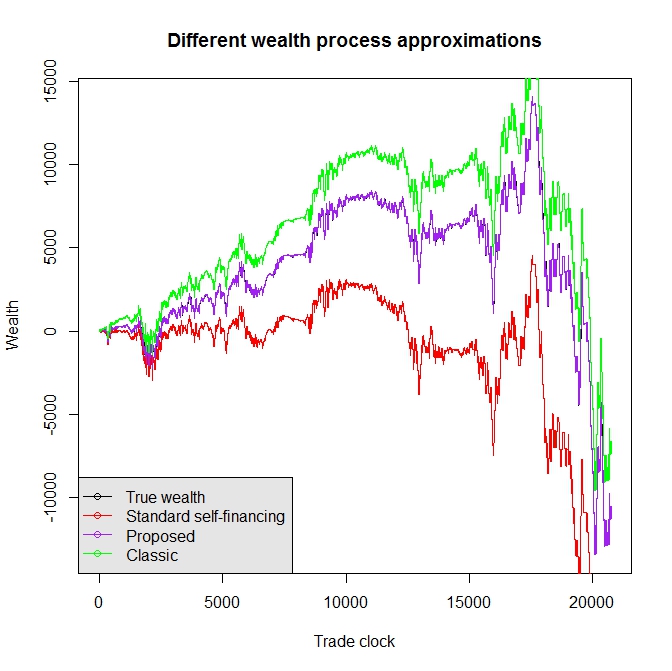}
	\caption{Example: Coca Cola on 04/18/13. Plots of the actual wealth of the aggregate passive trader together with the wealth computed from the three self-financing conditions. Red is the frictionless case \eqref{wealth_frictionless}. Green corresponds to  \eqref{wealth_transaction_costs}. 
	The actual wealth and the wealth computed from our self-financing condition (\ref{wealth_complete}) are indistinguishable on the graph. }
	\label{fig:wealth_approx}
\end{figure}

We first compare the impact of transaction costs and price impact by computing wealth with each of the three equations: the frictionless case given in \eqref{wealth_frictionless}, the case with transaction costs given in \eqref{wealth_transaction_costs} and the exact wealth equation \eqref{wealth_complete}. For the stocks of our universe, we provide the results in Table 3. 

\begin{center}
\tiny{

		\begin{longtable}{|l|c|c|c|}
			stock symbol & relative error & friction ratio & net P\&L  \\
				\hline
					 AA       & 177\% & 8\% & 5970 \$ \\ 
					 AAPL     & 9\% & 19\% & -1779103 \$ \\ 
					 ADBE     & 49\% & 44\% & -5746 \$ \\ 
					 AGN      & 564\% & 78\% & -28 \$ \\ 
					 AINV     & 140\% & 16\% & 1062 \$ \\ 
					 AMAT     & 80\% & 7\% & -15806 \$ \\ 
					 AMED     & 173\% & 40\% & 181 \$ \\ 
					 AMGN     & 15\% & 46\% & -55395 \$ \\ 
					 AMZN     & 63\% & 30\% & -57491 \$ \\ 
					 ANGO     & 149\% & 44\% & 127 \$ \\ 
					 APOG     & 93\% & 45\% & 947 \$ \\ 
					 ARCC     & 54\% & 21\% & 2690 \$ \\ 
					 AXP      & 19\% & 41\% & -37162 \$ \\ 
					 AYI      & 266\% & 64\% & 495 \$ \\ 
					 AZZ      & 162\% & 50\% & 57 \$ \\ 
					 BAS      & 1\% & 96\% & -1594 \$ \\ 
					 BHI      & 6\% & 61\% & 44602 \$ \\ 
					 BIIB     & 51\% & 40\% & -19266 \$ \\ 
					 BRCM     & 41\% & 29\% & 23826 \$ \\ 
					 BRE      & 8\% & 58\% & -3772 \$ \\ 
					 BXS      & 108\% & 69\% & -133 \$ \\ 
					 BZ       & 18\% & 85\% & -208 \$ \\ 
					 CB       & 58\% & 58\% & -1417 \$ \\ 
					 CBEY     & 81\% & 41\% & 399 \$ \\ 
					 CBT      & 25\% & 74\% & 769 \$ \\ 
					 CBZ      & 1235\% & 67\% & -21 \$ \\ 
					 CDR      & 13\% & 62\% & -260 \$ \\ 
					 CELG     & 55\% & 48\% & -12403 \$ \\ 
					 CETV     & 122\% & 40\% & 108 \$ \\ 
					 CKH      & 24\% & 58\% & -1248 \$ \\ 
					 CMCSA    & 46\% & 22\% & 25599 \$ \\ 
					 CNQR     & 13\% & 57\% & 5837 \$ \\ 
					 COO      & 33\% & 53\% & 1332 \$ \\ 
					 COST     & 7\% & 51\% & -32206 \$ \\ 
					 CPSI     & 44\% & 76\% & 540 \$ \\ 
					 CPWR     & 256\% & 23\% & 403 \$ \\ 
					 CR       & 32\% & 70\% & -1619 \$ \\ 
					 CRI      & 101\% & 76\% & -302 \$ \\ 
					 CSCO     & 193\% & 9\% & 15086 \$ \\ 
					 CSE      & 294\% & 27\% & 220 \$ \\ 
					 CSL      & 128\% & 60\% & -337 \$ \\ 
					 CTRN     & 106\% & 38\% & 162 \$ \\
					 CTSH     & 18\% & 51\% & 16447 \$ \\ 
					 DCOM     & 54\% & 71\% & 206 \$ \\ 
					 DELL     & 133\% & 7\% & -7695 \$ \\ 
					 DIS      & 25\% & 38\% & 21363 \$ \\ 
					 DK       & 28\% & 65\% & 2580 \$ \\ 
					 DOW      & 394\% & 35\% & 921 \$ \\ 
					 EBAY     & 406\% & 34\% & 6854 \$ \\ 
					 ESRX     & 10\% & 47\% & 34609 \$ \\ 
					 EWBC     & 9\% & 36\% & 31326 \$ \\ 
					 FCN      & 122\% & 65\% & 101 \$ \\ 
					 FFIC     & 48\% & 25\% & 418 \$ \\ 
					 FL       & 11\% & 50\% & -9114 \$ \\ 
					 FMER     & 35\% & 41\% & -4522 \$ \\ 
					 FPO      & 56\% & 114\% & -85 \$ \\ 
					 FRED     & 310\% & 51\% & 85 \$ \\ 
					 FULT     & 41\% & 25\% & -4251 \$ \\ 
					 GAS      & 65\% & 62\% & 325 \$ \\ 
					 GE       & 382\% & 11\% & -4446 \$ \\ 
					 GILD     & 14\% & 39\% & -57605 \$ \\ 
					 GLW      & 92\% & 6\% & -12881 \$ \\ 
					 GOOG     & 22\% & 32\% & -211247 \$ \\ 
					 GPS      & 61\% & 48\% & 3788 \$ \\ 
					 HON      & 10\% & 70\% & 13206 \$ \\ 
					 HPQ      & 73\% & 19\% & -18421 \$ \\ 
					 IMGN     & 35\% & 32\% & -5303 \$ \\ 
					 INTC     & 61\% & 5\% & 111422 \$ \\ 
					 IPAR     & 57\% & 22\% & -316 \$ \\ 
					 ISIL     & 14680\% & 28\% & -8 \$ \\ 
					 ISRG     & 320\% & 45\% & -6180 \$ \\ 
					 JKHY     & 19\% & 56\% & -2189 \$ \\ 
					 KMB      & 0\% & 74\% & 26968 \$ \\ 
					 KR       & 310\% & 36\% & -839 \$ \\ 
					 LANC     & 110\% & 33\% & 466 \$ \\ 
					 LECO     & 213\% & 60\% & 273 \$ \\ 
					 LPNT     & 56\% & 47\% & 2514 \$ \\ 
					 LSTR     & 357\% & 57\% & 272 \$ \\ 
					 MAKO     & 42\% & 41\% & 2031 \$ \\ 
					 MANT     & 53\% & 60\% & 507 \$ \\ 
					 MDCO     & 8\% & 56\% & 8883 \$ \\ 
					 MELI     & 17\% & 36\% & 14901 \$ \\
					 MIG      & 62\% & 26\% & 354 \$ \\ 
					 MMM      & 66\% & 47\% & -4133 \$ \\ 
					 MOD      & 61\% & 75\% & -41 \$ \\ 
					 MOS      & 13\% & 60\% & -10969 \$ \\ 
					 MRTN     & 45\% & 42\% & -845 \$ \\ 
					 MXWL     & 191276\% & 63\% & 0 \$ \\ 
					 NSR      & 33\% & 67\% & 975 \$ \\ 
					 NUS      & 10\% & 42\% & 22755 \$ \\ 
					 NXTM     & 275\% & 40\% & 76 \$ \\ 
					 PBH      & 27\% & 43\% & 436 \$ \\ 
					 PFE      & 51\% & 15\% & 26718 \$ \\ 
					 PG       & 2356\% & 34\% & 337 \$ \\ 
					 PNC      & 25\% & 65\% & -6350 \$ \\ 
					 PNY      & 16\% & 58\% & 1485 \$ \\ 
					 PTP      & 47\% & 69\% & 803 \$ \\ 
					 RIGL     & 54\% & 29\% & -2456 \$ \\ 
					 ROC      & 27\% & 58\% & 3124 \$ \\ 
					 ROCK     & 162\% & 38\% & 290 \$ \\ 
					 SF       & 8\% & 84\% & -1027 \$ \\ 
					 SFG      & 28\% & 57\% & 572 \$ \\ 
					 SWN      & 37\% & 45\% & -12890 \$ 
	\end{longtable}
{Table 2. Relative mispricing of wealth using the fricitonless wealth equation. Ratio between adverse selection and transaction cost friction terms. End of day gains of the aggregate passive trades. The table was computed for the date 04/18/13.}
	
}
\end{center}

For the sake of illustration, we also give a plot in figure \ref{fig:wealth_approx} for our example stock, Coca Cola on 04/18/13. The plot in particular shows how current models of wealth cannot even track the realized wealth \emph{ex-post}! While microsctructure noise models such as \cite{Yacine_noise} can be used to measure the error introduced by ignoring the high frequency microstructure, an exact reconstruction of wealth that matches the data perfectly is possible and not more difficult.

As expected, the frictionless equation underestimates true wealth while the equation ignoring price impact over-estimates the wealth of our aggregate liquidity provider. The relative error is always significant. This comes as no surprise as microstructure noise adds up very quickly over a trading day. 

\section{Applications}\label{sec:applications}
Our approach stems from a precise analysis of three ingredients: the information structure of the agents present in high frequency markets, the rules underpinning trades on the electronic exchanges and the vast amount of data present in these markets. As a result, it is particularly apt at answering practical questions both from a regulatory and a trading perspective and grounding the research in the modern world.

\subsection{Measuring HFT profits}
Monitoring the profits made by high frequency traders from their superior information is a first step towards their regulation. Some HFT traders also employ passive strategies that are not covered by our model.

First, let us compute the cumulative losses due to adverse selection, the third term of our wealth equation:
\begin{equation}
\sum_{n=1}^N \Delta_n p\Delta_n L
\end{equation}
We give the plot of this quantity as a function of $N$ in figure \ref{fig:cross-section} for a selection of stocks. Each path corresponds to the cumulative losses of the whole passive side of the market to traders with superior information for a given stock on a single day, rescaled. Note that the path is perfectly decreasing, in line with our model for superior information. In a standard informed trader model such as \cite{Kyle}, these paths would be random walks.

\begin{figure}
	\centering
		\includegraphics[width=1.00\textwidth]{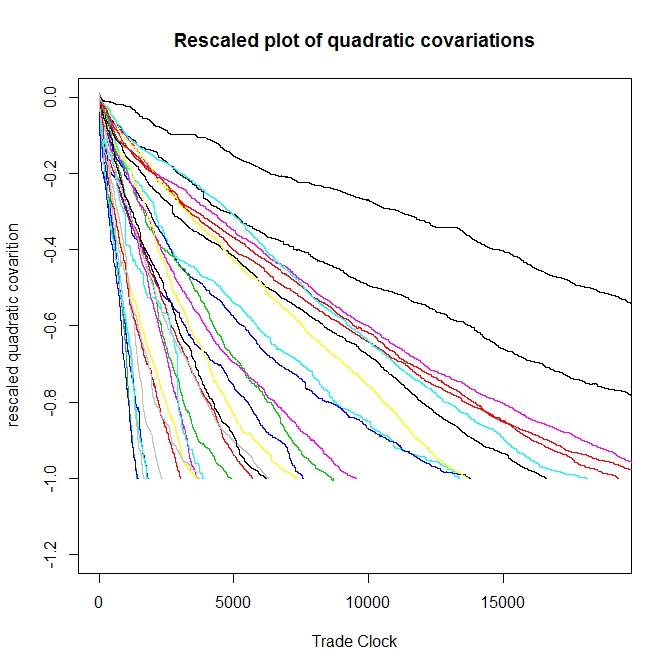}
	\caption{Empirical cumulative losses due to adverse selection (rescaled) for selected stocks.}
	\label{fig:cross-section}
\end{figure}

To get a sense of whether passive traders are -on average- sufficiently compensated for these losses, one can compare the transaction costs earned to the money lost to adverse selection.

\begin{figure}[htbp]
	\centering
		\includegraphics[width=1.00\textwidth]{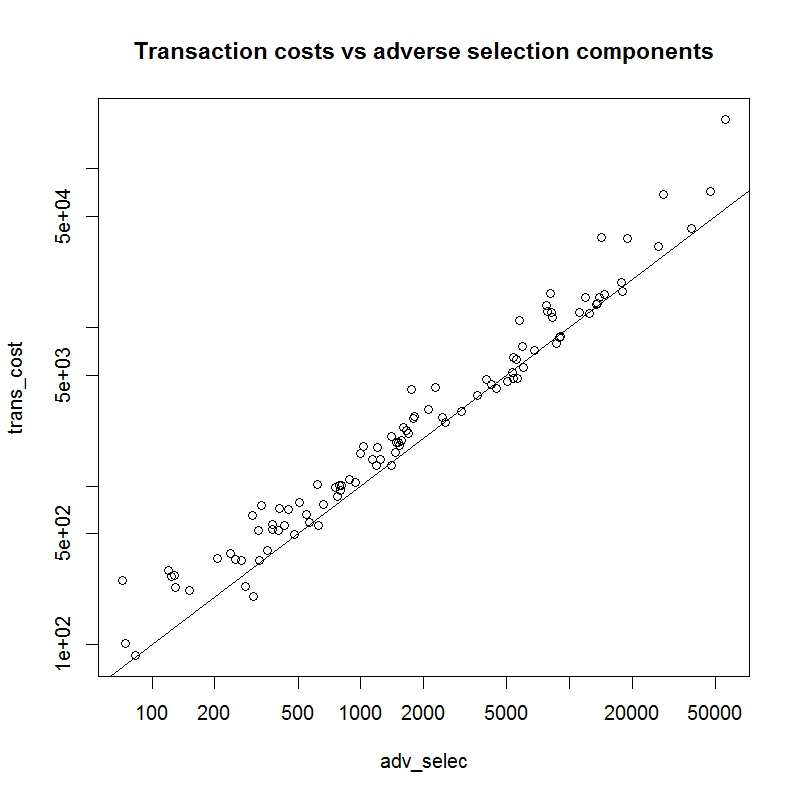}
		\caption{Transaction cost against adverse selection components for multiple stocks. Line corresponds to $y=x$. Plot is in log-log scale and each point corresponds to one stock on a given day.}
	\label{fig:component_comp}
\end{figure}

As the graph shows, for nearly all the stocks, transaction cost gains are larger than losses to adverse selection. This means that there is a sufficient number of noise traders to support the activity of the passive traders. However, a linear regression shows that the effective spread captured by these passive traders amounts to $20\%$ of the posted spread. The other $80\%$ go to the traders with superior information. This is a form of information rent.

\subsection{An econometric test for adverse selection}

We present an econometric test for adverse selection of the kind described above. We assume that both the midprice and the aggregate passive trader's inventory are It\^o processes, say of the form:
\begin{equation}
\begin{cases}
dp_t = \mu_t dt + \sigma_t dW_t \\
dL_t = b_t dt + l_t dW'_t
\end{cases}
\end{equation}
with instantaneous correlation $\rho_t$ between $W_t$ and $W'_t$. Then a continuous version of the adverse selection model implies that 
\begin{equation}
\rho_t \le 0,
\end{equation}
that is, the changes of a passive trader's inventory are always negatively correlated to the price returns. We test the following \emph{null hypothesis}:

{$\bf{H_0}$: There exists $t \in [0,1]$ such that $\rho_t >0$.}

If we denote by $p^N$ and $L^N$ the discrete measurements of the diffusion processes $p$ and $L$ on the uniform grid $\{1/N, 2/N, ..., 1\}$ then \cite{Yacine} suggests to consider:
\begin{equation}
	\begin{cases}
		C_t &= \sum_{n=1}^{\lfloor Nt\rfloor -1} \Delta_n p^N \Delta_n L^N \\
		V^N_t &= N \sum_{n=1}^{\lfloor Nt\rfloor -2} \left(  \left(\Delta_n p^N \Delta_{n+1} L^N\right)^2 + \Delta_n p^N \Delta_n L^N \Delta_{n+1} p^N \Delta_{n+1} L^N \right),
	\end{cases}
\end{equation}
which leads to the following form of the central limit theorem
\begin{equation}
\mathcal{L}\left( \frac{C^N_t - [p,L]_t}{\sqrt{N^{-1} |V^N_t|}} \right) \rightarrow N(0,1).
\end{equation}
The result can then be used to compute rejection probabilities of the null hypothesis on buckets of $M$ trades assuming that $\rho_t$, $\sigma_t$ and $l_t$ are constant on these buckets. We then multiply those rejection probabilities to obtain the overall rejection probability for the null hypothesis. The results are given in Table 2 in the case where $M$ is chosen for each stock so that the trading day is divided into 8 buckets.

\begin{center}
\tiny{

		\begin{longtable}{|l c|l c|l c|}
			stock symbol & prob rejection & stock symbol & prob rejection & stock symbol & prob rejection \\
				\hline
					 AA       & 0.45458 &  CPSI     & 0.99729 &  IPAR     & 0.86019 \\ 
					 AAPL     & 1 &  CPWR     & 0.99433 &  ISIL     & 0.90428 \\ 
					 ADBE     & 1 &  CR       & 0.99835 &  ISRG     & 0.99999 \\ 
					 AGN      & 0.99564 &  CRI      & 0.99438 &  JKHY     & 0.97535 \\ 
					 AINV     & 0.66919 &  CSCO     & 0.99975 &  KMB      & 0.99999 \\ 
					 AMAT     & 0.88248 &  CSE      & 0.93149 &  KR       & 0.99999 \\ 
					 AMED     & 0.92732 &  CSL      & 0.99042 &  LANC     & 0.99475 \\ 
					 AMGN     & 0.99999 &  CTRN     & 0.99998 &  LECO     & 0.99997 \\ 
					 AMZN     & 0.99999 &  CTSH     & 0.99999 &  LPNT     & 0.99874 \\ 
					 ANGO     & 0.99496 &  DCOM     & 0.98683 &  LSTR     & 0.99994 \\ 
					 APOG     & 0.99973 &  DELL     & 0.67247 &  MAKO     & 0.75563 \\ 
					 ARCC     & 0.93312 &  DIS      & 0.99999 &  MANT     & 0.99912 \\ 
					 AXP      & 0.99999 &  DK       & 0.99988 &  MDCO     & 0.99999 \\ 
					 AYI      & 0.99053 &  DOW      & 0.99999 &  MELI     & 0.84372 \\ 
					 AZZ      & 0.99754 &  EBAY     & 1 &  MMM      & 0.98094 \\ 
					 BAS      & 0.99952 &  ESRX     & 1 &  MOD      & 0.99999 \\ 
					 BHI      & 1 &  EWBC     & 0.99999 &  MOS      & 0.99999 \\ 
					 BIIB     & 0.99999 &  FCN      & 0.99974 &  MRTN     & 0.99941 \\ 
					 BRCM     & 1 &  FFIC     & 0.99922 &  MXWL     & 0.99581 \\ 
					 BRE      & 0.9981 &  FL       & 0.99999 &  NSR      & 0.81248 \\ 
					 BXS      & 0.9992 &  FMER     & 0.99911 &  NUS      & 0.99385 \\ 
					 BZ       & 0.9997 &  FPO      & 0.99979 &  NXTM     & 0.9954 \\ 
					 CB       & 0.99999 &  FRED     & 0.99916 &  PBH      & 0.98192 \\ 
					 CBEY     & 0.66095 &  FULT     & 0.91296 &  PFE      & 0.99999 \\ 
					 CBT      & 0.98696 &  GAS      & 0.89222 &  PG       & 1 \\ 
					 CBZ      & 0.99992 &  GE       & 0.99989 &  PNC      & 0.99999 \\ 
					 CDR      & 0.99999 &  GILD     & 0.99999 &  PNY      & 0.98274 \\ 
					 CELG     & 0.99999 &  GLW      & 0.56781 &  PTP      & 0.99217 \\ 
					 CETV     & 0.99917 &  GOOG     & 0.99999 &  RIGL     & 0.98968 \\ 
					 CKH      & 0.99832 &  GPS      & 0.99999 &  ROC      & 0.99661 \\ 
					 CMCSA    & 0.99999 &  HON      & 0.99999 &  ROCK     & 0.87984 \\ 
					 CNQR     & 0.99822 &  HPQ      & 0.99999 &  SF       & 0.99951 \\ 
					 COO      & 0.97052 &  IMGN     & 0.99232 &  SFG      & 0.99929 \\ 
					 COST     & 1 &  INTC     & 0.99113 &  SWN      & 1 \\ 

					\end{longtable}
	{Table 3. Probability of rejection for the null hypothesis on 04/18/13.}
	
}
\end{center}

Table 3 rejects the null hypothesis for nearly all stocks. As a conclusion, instantaneous adverse selection is very easily measured and modeled on the continuous time scale. We only tested for the \emph{sign} of the adverse selection coefficient $\rho$, but more practical-minded researchers may be interested in its magnitude. The correlation is around $-0.3$ across most stocks. This implies that only $30\%$ of the inventory of passive traders is due to adverse selection, but $80\%$ of their profits go to traders with superior information.

\subsection{Implications for option pricing}
The main result of this rather long and technical subsection is a result stating which options can be hedged via limit orders and which can only be hedged via market orders in a Black and Scholes world.

For tractability reasons, it is often convenient to work in the framework of continuous time finance. Some quantities, notably the midprice, can then be modeled as having a Brownian motion component. Using a diffusion approximation argument, we derive the analogue of our wealth equation in continuous time, and revisit the classic Black - Scholes option hedging argument with this corrected equation.

\subsubsection{The limiting equation}

Recall the discrete wealth equation:
\begin{equation}
\Delta_n X = L_n \Delta_n p \pm \frac{s_n}{2}\left|\Delta_n L\right| + \Delta_n p \Delta_n L. \label{disc_self_financing}
\end{equation}

We start with the most standard hypothesis: assuming that we observe a continuous, diffusive price at discrete times. This leads to a continuous time model of the form:
\begin{equation}
dp_t = \mu_t dt + \sigma_t dW_t.
\end{equation}

We will also model the inventory with a continuous diffusion process:
\begin{equation}
dL_t = b_t dt + l_t dW'_t
\end{equation}
with instantaneous correlation $\rho_t$ between $W_t$ and $W'_t$. Our rational for allowing inventories with a correlated Brownian motion is the relative importance of the adverse selection. Without this correlation between changes in inventories and prices, our model for the price would not provide any of the significant instantaneous adverse selection found in the data.

\begin{theorem}
The correct continuous-time clearing equation that incorporates both transaction costs and price impact is:
\begin{equation}
dX_t = L_t dp_t + \left(\rho_t \sigma_t \pm \frac{s_t}{\sqrt{2\pi}}\right)l_t dt \label{cont_self_financing}
\end{equation}
where $\pm$ is the sign of $\rho_t$.
\end{theorem}
In the limit case $\rho_t=0$, the trade was most likely done with a market order and hence $\pm$ is negative. The proof can be found in the appendix.

Note that if we allow $l_t$ to be signed and fix $\rho_t>0$, then the formula collapses to:
\begin{equation}
dX_t = L_t dp_t + \left(\rho_t \sigma_t - \frac{s_t}{\sqrt{2\pi}}\right)l_t dt \label{eq:continuous_signed}
\end{equation}
which is convenient when juggling between limit and market orders.

\subsubsection{Revisiting Black - Scholes}
In this section we assume that any inventory of the form
\begin{equation}
dL_t = b_t dt + l_t dW_t
\end{equation}
with $l_t$ \emph{signed}, is attainable. If $l_t\ge 0$, trading is done with market orders while for $l_t<0$ trades are done with limit orders.

Assume that the midprice follows a geometric Brownian motion as in the standard Black - Scholes model:
\begin{equation}
dp_t = \mu p_t dt + \sigma p_t dW_t,
\end{equation}
and that the spread is proportional to price volatility,
\begin{equation}
s_t = s \sigma p_t.
\end{equation}
for some constant $s$

If the interest rate $r$ is constant and there are no dividends, then this market is complete and we can price a European option using the standard no-arbitrage argument, albeit with a slightly different partial differential equation (PDE) because of the two sources of frictions. The pricing PDE for the value $v$ of a European option is
\begin{equation}
\frac{\partial v}{\partial t}(t,p) + \frac{\sigma^2 p^2}{2}\left(\sqrt{\frac{2}{\pi}}s - 1\right)\frac{\partial^2 v}{\partial p^2}(t,p) + r p \frac{\partial v}{\partial p}(t,p) = r v(t,p) \label{pricing_PDE}
\end{equation} 
which is the same equation as a local volatility model with a factor of
\begin{equation}
\sqrt{\sqrt{\frac{2}{\pi}}s - 1}
\end{equation}
on the volatility. The spread therefore acts as a multiplier on the implied volatility of an option. The larger the spread, the larger the implied volatility market makers will quote.  The full derivation is presented in the appendix.

A consequence of the proof is that the traditional delta-hedge still leads to perfect replication in a Black - Scholes market with transaction costs and instantaneous adverse selection. However, if we write the delta hedge equation
\begin{equation}
L_T = \frac{\partial v}{\partial p}(t,p_t),
\end{equation}
then that leads to the identity 
\begin{equation}
l_t = \sigma p_t \frac{\partial^2 v}{\partial p^2}(t,p_t),
\end{equation}
which in particular implies that delta-hedging \emph{negative gamma} options can only be done with \emph{limit orders}, while delta-hedging \emph{positive gamma} options is done with \emph{market orders}.

For an intuitive example of this phenomenon consider a call option. It has positive gamma: in order to delta-hedge this option, one needs to buy when the price goes up and sell when the price goes down. But, because of adverse selection, the price tends to move down when you buy with limit orders and up when you sell with limit orders. Therefore, market orders are to be preferred to delta-hedge a call option. However, if you are selling a call option then the opposite of a call option must be delta-hedged, which can be done with limit orders.

In addition to computing prices and delta-hedging ratios under transaction costs and instantaneous adverse selection, this theory suggests an execution strategy by quantifying when limit or market orders should be used.

\appendix

\section{NASDAQ data}

Our raw data consists of NASDAQ's ITCH data files, which include all visible limit and market orders, as well as market orders that execute hidden liquidity. We clean up the data by removing 'special deals', trades that happen within the bid-ask spread, and market orders that execute hidden liquidity. This leaves us with only visible market orders hitting visible limit orders at the best bid or ask price. While this is the natural state of the data, one may wish to reconstitute parent orders, in which case not all trades would happen at the best bid or ask price. We propose a method for doing so later in the appendix.

Because the data does not provide the identity of the different traders, we work by aggregating all the passive trades. This means we parse all the trades as viewed from the limit order's perspective, construct the aggregate portfolio, and consider that as our 'trader'. The figures included in this paper were produced using the data for Coca Cola (KO) on 18/04/13. A full cross-section of 120 stocks used in the recent ECB study \cite{ECB} was considered over multiple days to test for robustness of our results. We removed some of the stocks from certain tables when the number of trades was not sufficiently large to run the corresponding analysis.

\begin{figure}[htbp]
	\centering
		\includegraphics[width=1.00\textwidth]{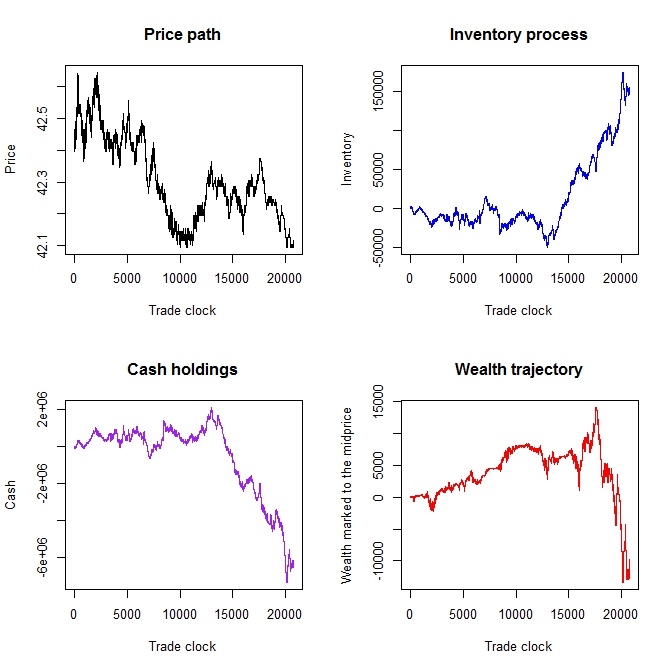}
	\caption{Example : Coca Cola (KO) stock on 18/04/13. Inventory, cash and wealth are those of the aggregate liquidity provider.}
	\label{fig:Empirical_Summary}
\end{figure}

\subsection{Ignored trades}
For the sake of completeness, we include a table of how many executions were discarded for each stock. They correspond to symbols 'C' and 'P' in the NASDAQ ITCH messaging convention. Type 'C' messages correspond to special deals while type 'P' messages to trades that execute against a hidden limit order. Note that -for some stocks- hidden executions are very important.

\begin{center}
\tiny{

  \begin{longtable}{|l c c|l c c |}
			stock symbol & percent type 'C' & percent type 'P' & stock symbol & percent type 'C' & percent type 'P'  \\
				\hline 
					 AA       & 1.4\% & 16.6\% &  FCN      & 0.3\% & 24\% \\ 
					 AAPL     & 0.4\% & 36.3\% &  FFIC     & 0\% & 24.2\% \\ 
					 ADBE     & 1.3\% & 7.5\% &  FL       & 1.4\% & 6.2\% \\ 
					 AGN      & 1.2\% & 12.9\% &  FMER     & 1.5\% & 8.4\% \\ 
					 AINV     & 1.4\% & 18.8\% &  FPO      & 0.2\% & 24.1\% \\ 
					 AMAT     & 1.8\% & 14.2\% &  FRED     & 0.3\% & 7.4\% \\ 
					 AMED     & 0.7\% & 13.6\% &  FULT     & 1.6\% & 8.5\% \\ 
					 AMGN     & 1.5\% & 11.8\% &  GAS      & 1.1\% & 18.3\% \\ 
					 AMZN     & 0.5\% & 25.9\% &  GE       & 2.7\% & 9.1\% \\ 
					 ANGO     & 0.3\% & 19\% &  GILD     & 1.4\% & 9.5\% \\ 
					 APOG     & 0.5\% & 15.2\% &  GLW      & 1.6\% & 13.4\% \\ 
					 ARCC     & 4.5\% & 9.9\% &  GOOG     & 0.9\% & 31\% \\ 
					 AXP      & 2.1\% & 13.8\% &  GPS      & 2.2\% & 7.8\% \\ 
					 AYI      & 0.7\% & 26.8\% &  HON      & 2.5\% & 13.6\% \\ 
					 AZZ      & 0.8\% & 17.6\% &  HPQ      & 2.3\% & 5.1\% \\ 
					 BAS      & 0.9\% & 12.2\% &  IMGN     & 0.8\% & 13.4\% \\ 
					 BHI      & 2.8\% & 8.2\% &  INTC     & 2.2\% & 11.9\% \\ 
					 BIIB     & 0.8\% & 23.7\% &  IPAR     & 0\% & 25.5\% \\ 
					 BRCM     & 2.6\% & 6.7\% &  ISIL     & 1\% & 11.3\% \\ 
					 BRE      & 1.6\% & 14.3\% &  ISRG     & 0.3\% & 27.3\% \\ 
					 BXS      & 5.3\% & 15.7\% &  JKHY     & 0.8\% & 15.1\% \\ 
					 BZ       & 4.9\% & 12.1\% &  KMB      & 1.7\% & 17.3\% \\ 
					 CB       & 1.6\% & 11.2\% &  KR       & 1.7\% & 8.1\% \\ 
					 CBEY     & 0.5\% & 8.7\% &  LANC     & 1\% & 23.4\% \\ 
					 CBT      & 0.4\% & 20.1\% &  LECO     & 0.3\% & 12.2\% \\ 
					 CBZ      & 1.4\% & 8.8\% &  LPNT     & 0.6\% & 20.5\% \\ 
					 CDR      & 2.9\% & 15.2\% &  LSTR     & 0.1\% & 12.2\% \\ 
					 CELG     & 0.8\% & 25\% &  MAKO     & 1.6\% & 12\% \\ 
					 CETV     & 0\% & 5.3\% &  MANT     & 0.3\% & 15.2\% \\ 
					 CKH      & 1.4\% & 19.1\% &  MDCO     & 0.6\% & 15\% \\ 
					 CMCSA    & 1.7\% & 10.6\% &  MELI     & 0.2\% & 28.9\% \\ 
					 CNQR     & 0.8\% & 16.3\% &  MIG      & 2\% & 15.9\% \\ 
					 COO      & 0.3\% & 23.9\% &  MMM      & 1.7\% & 14.4\% \\ 
					 COST     & 1.3\% & 6.4\% &  MOD      & 0.7\% & 16.5\% \\ 
					 CPSI     & 0.3\% & 41.1\% &  MOS      & 2\% & 11.2\% \\ 
					 CPWR     & 2.2\% & 8.9\% &  MRTN     & 0\% & 18.5\% \\ 
					 CR       & 0.6\% & 9.3\% &  MXWL     & 1.1\% & 14\% \\ 
					 CRI      & 0.6\% & 12.9\% &  NSR      & 1.4\% & 17.7\% \\ 
					 CSCO     & 1.1\% & 8.9\% &  NUS      & 1.7\% & 15.8\% \\ 
					 CSE      & 1.2\% & 16.7\% &  NXTM     & 1.6\% & 21.9\% \\ 
					 CSL      & 0.5\% & 27\% &  PBH      & 0.7\% & 11.4\% \\ 
					 CTRN     & 1.5\% & 23.8\% &  PFE      & 2.7\% & 9.7\% \\ 
					 CTSH     & 1.7\% & 12.2\% &  PG       & 3.1\% & 7\% \\ 
					 DCOM     & 0.8\% & 16\% &  PNC      & 1.5\% & 11.8\% \\ 
					 DELL     & 3.2\% & 12\% &  PNY      & 1\% & 8\% \\ 
					 DIS      & 2.2\% & 6.9\% &  PTP      & 0.4\% & 15.4\% \\ 
					 DK       & 1.8\% & 17.9\% &  RIGL     & 1\% & 12.4\% \\ 
					 DOW      & 2.1\% & 6.7\% &  ROC      & 1.7\% & 11.7\% \\ 
					 EBAY     & 2.1\% & 11.6\% &  ROCK     & 0.8\% & 8.2\% \\ 
					 ESRX     & 2.1\% & 8.5\% &  SF       & 0.3\% & 14.8\% \\ 
					 EWBC     & 2.1\% & 9.8\% &  SFG      & 1.6\% & 50.7\% \\ 
					          &       &       &  SWN      & 2.3\% & 9.8\% 
					\end{longtable}
	{Table 4. Proportion of ignored trades. Type `C' correspond to `special deals' while type `P' execute against hidden limit orders.}
}
\end{center}

\subsection{Reconstructing parent orders}

Throughout the paper we assume that all trades happen at the best bid or ask price. In general however, a small percentage of trades do go deeper into the book and consume liquidity from the next bid and ask prices. Unfortunately, the raw ITCH data does not indicate how to reconstruct these `parent orders' from the sequence of messages printed in the file. We propose a reconstruction heuristic and, for the sake of completeness, present our empirical results in the general case.

\begin{center}
\tiny{
\begin{longtable}{|l|c|c|c|c|}
			stock symbol & total nb of trades & with price impact & without price impact & reverse price impact  \\
				\hline
					  AA       & 1330 & 324 (24.36\%) & 1005 & 1 (0.08\%) \\ 
					 AAPL     & 29496 & 15979 (54.17\%) & 11194 & 2323 (7.88\%) \\ 
					 ADBE     & 4724 & 2250 (47.62\%) & 2326 & 148 (3.14\%) \\ 
					 AGN      & 1523 & 1053 (69.13\%) & 314 & 156 (10.25\%) \\ 
					 AINV     & 785 & 198 (25.22\%) & 582 & 5 (0.64\%) \\ 
					 AMAT     & 2632 & 522 (19.83\%) & 2105 & 5 (0.19\%) \\ 
					 AMED     & 495 & 195 (39.39\%) & 276 & 24 (4.85\%) \\ 
					 AMGN     & 9832 & 6369 (64.77\%) & 2918 & 545 (5.55\%) \\ 
					 AMZN     & 9101 & 5746 (63.13\%) & 2717 & 638 (7.02\%) \\ 
					 ANGO     & 341 & 128 (37.53\%) & 196 & 17 (4.99\%) \\ 
					 APOG     & 675 & 407 (60.29\%) & 213 & 55 (8.15\%) \\ 
					 ARCC     & 926 & 246 (26.56\%) & 676 & 4 (0.44\%) \\ 
					 AXP      & 8147 & 3985 (48.91\%) & 3814 & 348 (4.28\%) \\ 
					 AYI      & 663 & 437 (65.91\%) & 163 & 63 (9.51\%) \\ 
					 AZZ      & 102 & 63 (61.76\%) & 21 & 18 (17.65\%) \\ 
					 BAS      & 393 & 268 (68.19\%) & 87 & 38 (9.67\%) \\ 
					 BHI      & 5966 & 3444 (57.72\%) & 2165 & 357 (5.99\%) \\ 
					 BIIB     & 3633 & 2256 (62.09\%) & 1055 & 322 (8.87\%) \\ 
					 BRCM     & 6849 & 3269 (47.72\%) & 3452 & 128 (1.87\%) \\ 
					 BRE      & 896 & 486 (54.24\%) & 340 & 70 (7.82\%) \\ 
					 BXS      & 495 & 272 (54.94\%) & 200 & 23 (4.65\%) \\ 
					 BZ       & 183 & 110 (60.1\%) & 51 & 22 (12.03\%) \\ 
					 CB       & 1393 & 919 (65.97\%) & 371 & 103 (7.4\%) \\ 
					 CBEY     & 404 & 168 (41.58\%) & 200 & 36 (8.92\%) \\ 
					 CBT      & 527 & 327 (62.04\%) & 132 & 68 (12.91\%) \\ 
					 CBZ      & 100 & 45 (45\%) & 47 & 8 (8\%) \\ 
					 CDR      & 113 & 58 (51.32\%) & 51 & 4 (3.54\%) \\ 
					 CELG     & 5602 & 3737 (66.7\%) & 1332 & 533 (9.52\%) \\ 
					 CETV     & 157 & 72 (45.85\%) & 75 & 10 (6.37\%) \\ 
					 CKH      & 192 & 122 (63.54\%) & 44 & 26 (13.55\%) \\ 
					 CMCSA    & 6410 & 2402 (37.47\%) & 3969 & 39 (0.61\%) \\ 
					 CNQR     & 769 & 465 (60.46\%) & 238 & 66 (8.59\%) \\ 
					 COO      & 202 & 125 (61.88\%) & 51 & 26 (12.88\%) \\ 
					 COST     & 5473 & 3483 (63.63\%) & 1687 & 303 (5.54\%) \\ 
					 CPSI     & 297 & 164 (55.21\%) & 102 & 31 (10.44\%) \\ 
					 CPWR     & 686 & 273 (39.79\%) & 407 & 6 (0.88\%) \\ 
					 CR       & 493 & 298 (60.44\%) & 146 & 49 (9.94\%) \\ 
					 CRI      & 703 & 483 (68.7\%) & 151 & 69 (9.82\%) \\ 
					 CSCO     & 5025 & 1315 (26.16\%) & 3706 & 4 (0.08\%) \\ 
					 CSE      & 486 & 168 (34.56\%) & 316 & 2 (0.42\%) \\ 
					 CSL      & 338 & 187 (55.32\%) & 108 & 43 (12.73\%) \\ 
					 CTRN     & 112 & 62 (55.35\%) & 39 & 11 (9.83\%) \\ 
					 CTSH     & 4924 & 2941 (59.72\%) & 1770 & 213 (4.33\%) \\ 
					 DCOM     & 248 & 132 (53.22\%) & 84 & 32 (12.91\%) \\ 
					 DELL     & 1162 & 198 (17.03\%) & 963 & 1 (0.09\%) \\ 
					 DIS      & 5801 & 2900 (49.99\%) & 2762 & 139 (2.4\%) \\ 
					 DK       & 868 & 535 (61.63\%) & 219 & 114 (13.14\%) \\ 
					 DOW      & 3552 & 1679 (47.26\%) & 1807 & 66 (1.86\%) \\ 
					 EBAY     & 21407 & 11430 (53.39\%) & 9191 & 786 (3.68\%) \\ 
					 ESRX     & 5771 & 2957 (51.23\%) & 2595 & 219 (3.8\%) \\ 
					 EWBC     & 3000 & 1328 (44.26\%) & 1599 & 73 (2.44\%) \\ 
					 FCN      & 202 & 134 (66.33\%) & 53 & 15 (7.43\%) \\ 
					 FFIC     & 150 & 75 (50\%) & 58 & 17 (11.34\%) \\ 
					 FL       & 2183 & 1099 (50.34\%) & 1026 & 58 (2.66\%) \\ 
					 FMER     & 2015 & 751 (37.27\%) & 1233 & 31 (1.54\%) \\ 
					 FPO      & 207 & 145 (70.04\%) & 50 & 12 (5.8\%) \\ 
					 FRED     & 530 & 281 (53.01\%) & 217 & 32 (6.04\%) \\ 
					 FULT     & 1626 & 454 (27.92\%) & 1161 & 11 (0.68\%) \\ 
					 GAS      & 494 & 302 (61.13\%) & 132 & 60 (12.15\%) \\ 
					 GE       & 3351 & 913 (27.24\%) & 2434 & 4 (0.12\%) \\ 
					 GILD     & 9700 & 4707 (48.52\%) & 4704 & 289 (2.98\%) \\ 
					 GLW      & 1703 & 399 (23.42\%) & 1303 & 1 (0.06\%) \\ 
					 GOOG     & 6745 & 3766 (55.83\%) & 2366 & 613 (9.09\%) \\ 
					 GPS      & 3446 & 1809 (52.49\%) & 1540 & 97 (2.82\%) \\ 
					 HON      & 3622 & 2151 (59.38\%) & 1290 & 181 (5\%) \\ 
					 HPQ      & 4337 & 1819 (41.94\%) & 2498 & 20 (0.47\%) \\ 
					 IMGN     & 1461 & 705 (48.25\%) & 674 & 82 (5.62\%) \\ 
					 INTC     & 6863 & 1190 (17.33\%) & 5666 & 7 (0.11\%) \\ 
					 IPAR     & 138 & 66 (47.82\%) & 49 & 23 (16.67\%) \\ 
					 ISIL     & 1002 & 381 (38.02\%) & 612 & 9 (0.9\%) \\ 
					 ISRG     & 2015 & 1165 (57.81\%) & 654 & 196 (9.73\%) \\ 
					 JKHY     & 721 & 421 (58.39\%) & 240 & 60 (8.33\%) \\ 
					 KMB      & 2044 & 1394 (68.19\%) & 435 & 215 (10.52\%) \\ 
					 KR       & 2848 & 1251 (43.92\%) & 1566 & 31 (1.09\%) \\ 
					 LANC     & 168 & 99 (58.92\%) & 48 & 21 (12.5\%) \\ 
					 LECO     & 807 & 524 (64.93\%) & 219 & 64 (7.94\%) \\ 
					 LPNT     & 1401 & 821 (58.6\%) & 465 & 115 (8.21\%) \\ 
					 LSTR     & 938 & 571 (60.87\%) & 278 & 89 (9.49\%) \\ 
					 MAKO     & 756 & 337 (44.57\%) & 392 & 27 (3.58\%) \\ 
					 MANT     & 364 & 201 (55.21\%) & 124 & 39 (10.72\%) \\ 
					 MDCO     & 1219 & 750 (61.52\%) & 385 & 84 (6.9\%) \\ 
					 MELI     & 665 & 354 (53.23\%) & 240 & 71 (10.68\%) \\ 
					 MIG      & 302 & 71 (23.5\%) & 226 & 5 (1.66\%) \\ 
					 MMM      & 3212 & 2058 (64.07\%) & 874 & 280 (8.72\%) \\ 
					 MOS      & 2411 & 1541 (63.91\%) & 681 & 189 (7.84\%) \\ 
					 MRTN     & 307 & 141 (45.92\%) & 141 & 25 (8.15\%) \\ 
					 MXWL     & 327 & 175 (53.51\%) & 144 & 8 (2.45\%) \\ 
					 NSR      & 464 & 256 (55.17\%) & 157 & 51 (11\%) \\ 
					 NUS      & 1030 & 625 (60.67\%) & 295 & 110 (10.68\%) \\ 
					 NXTM     & 369 & 140 (37.94\%) & 205 & 24 (6.51\%) \\ 
					 PBH      & 105 & 57 (54.28\%) & 28 & 20 (19.05\%) \\ 
					 PFE      & 3729 & 1310 (35.13\%) & 2410 & 9 (0.25\%) \\ 
					 PG       & 7489 & 3742 (49.96\%) & 3557 & 190 (2.54\%) \\ 
					 PNC      & 3171 & 2000 (63.07\%) & 982 & 189 (5.97\%) \\ 
					 PNY      & 428 & 248 (57.94\%) & 119 & 61 (14.26\%) \\ 
					 PTP      & 462 & 291 (62.98\%) & 125 & 46 (9.96\%) \\ 
					 RIGL     & 1038 & 371 (35.74\%) & 652 & 15 (1.45\%) \\ 
					 ROC      & 811 & 521 (64.24\%) & 181 & 109 (13.45\%) \\ 
					 ROCK     & 531 & 214 (40.3\%) & 295 & 22 (4.15\%) \\ 
					 SF       & 367 & 241 (65.66\%) & 76 & 50 (13.63\%) \\ 
					 SFG      & 149 & 96 (64.42\%) & 29 & 24 (16.11\%) \\ 
					 SWN      & 6615 & 3604 (54.48\%) & 2787 & 224 (3.39\%) 
\end{longtable}
{Table 1-bis. Same table as Table 1 with reconstruction of parent trades.}
	
\begin{longtable}{|l|c|c|c|}
			stock symbol & relative error & friction ratio & net P\&L  \\
				\hline
					 AA       & 53\% & 72\% & 5952 \$ \\ 
					 AAPL     & 9\% & 27\% & -1776517 \$ \\ 
					 ADBE     & 1\% & 101\% & -5748 \$ \\ 
					 AGN      & 33\% & 101\% & 154 \$ \\ 
					 AINV     & 71\% & 58\% & 1049 \$ \\ 
					 AMAT     & 38\% & 55\% & -15810 \$ \\ 
					 AMED     & 111\% & 61\% & 181 \$ \\ 
					 AMGN     & 6\% & 76\% & -54822 \$ \\ 
					 AMZN     & 67\% & 38\% & -50457 \$ \\ 
					 ANGO     & 74\% & 72\% & 127 \$ \\ 
					 APOG     & 65\% & 61\% & 945 \$ \\ 
					 ARCC     & 14\% & 78\% & 2689 \$ \\ 
					 AXP      & 3\% & 89\% & -37167 \$ \\ 
					 AYI      & 270\% & 69\% & 339 \$ \\ 
					 AZZ      & 162\% & 50\% & 57 \$ \\ 
					 BAS      & 4\% & 111\% & -1594 \$ \\ 
					 BHI      & 1\% & 109\% & 44562 \$ \\ 
					 BIIB     & 43\% & 49\% & -19267 \$ \\ 
					 BRCM     & 2\% & 96\% & 23825 \$ \\ 
					 BRE      & 5\% & 79\% & -3915 \$ \\ 
					 BXS      & 2\% & 97\% & -134 \$ \\ 
					 BZ       & 19\% & 120\% & -210 \$ \\ 
					 CB       & 49\% & 80\% & -1420 \$ \\ 
					 CBEY     & 46\% & 66\% & 397 \$ \\ 
					 CBT      & 12\% & 86\% & 747 \$ \\ 
					 CBZ      & 660\% & 73\% & -5 \$ \\ 
					 CDR      & 23\% & 96\% & -264 \$ \\ 
					 CELG     & 46\% & 65\% & -12114 \$ \\ 
					 CETV     & 64\% & 68\% & 106 \$ \\ 
					 CKH      & 22\% & 62\% & -1278 \$ \\ 
					 CMCSA    & 6\% & 89\% & 25597 \$ \\ 
					 CNQR     & 9\% & 68\% & 5833 \$ \\ 
					 COO      & 35\% & 60\% & 1227 \$ \\ 
					 COST     & 5\% & 78\% & -33351 \$ \\ 
					 CPSI     & 43\% & 77\% & 534 \$ \\ 
					 CPWR     & 39\% & 88\% & 403 \$ \\ 
					 CR       & 5\% & 85\% & -1441 \$ \\ 
					 CRI      & 13\% & 103\% & -242 \$ \\ 
					 CSCO     & 37\% & 82\% & 15076 \$ \\ 
					 CSE      & 1513\% & 88\% & 270 \$ \\ 
					 CSL      & 95\% & 64\% & -246 \$ \\ 
					 CTRN     & 93\% & 46\% & 162 \$ \\ 
					 CTSH     & 4\% & 88\% & 16442 \$ \\ 
					 DCOM     & 18\% & 90\% & 205 \$ \\ 
					 DELL     & 70\% & 52\% & -7703 \$ \\ 
					 DIS      & 1\% & 103\% & 21357 \$ \\ 
					 DK       & 21\% & 79\% & 2473 \$ \\ 
					 DOW      & 43\% & 107\% & 916 \$ \\ 
					 EBAY     & 57\% & 91\% & 6931 \$ \\ 
					 ESRX     & 1\% & 94\% & 34603 \$ \\ 
					 EWBC     & 2\% & 85\% & 31333 \$ \\ 
					 FCN      & 73\% & 78\% & 119 \$ \\ 
					 FFIC     & 39\% & 45\% & 415 \$ \\ 
					 FL       & 0\% & 100\% & -9313 \$ \\ 
					 FMER     & 4\% & 91\% & -4523 \$ \\ 
					 FPO      & 73\% & 150\% & -125 \$ \\ 
					 FRED     & 186\% & 71\% & 84 \$ \\ 
					 FULT     & 17\% & 68\% & -4258 \$ \\ 
					 GAS      & 54\% & 76\% & 365 \$ \\ 
					 GE       & 35\% & 91\% & -4446 \$ \\ 
					 GILD     & 1\% & 95\% & -57616 \$ \\ 
					 GLW      & 37\% & 61\% & -12927 \$ \\ 
					 GOOG     & 11\% & 40\% & -211606 \$ \\ 
					 GPS      & 11\% & 109\% & 3787 \$ \\ 
					 HON      & 6\% & 117\% & 13242 \$ \\ 
					 HPQ      & 6\% & 106\% & -18432 \$ \\ 
					 IMGN     & 18\% & 65\% & -5328 \$ \\ 
					 INTC     & 22\% & 65\% & 111332 \$ \\ 
					 IPAR     & 52\% & 28\% & -317 \$ \\ 
					 ISIL     & 1473\% & 90\% & 11 \$ \\ 
					 ISRG     & 283\% & 52\% & -6182 \$ \\ 
					 JKHY     & 9\% & 78\% & -2198 \$ \\ 
					 KMB      & 10\% & 97\% & 25836 \$ \\ 
					 KR       & 43\% & 108\% & -846 \$ \\ 
					 LANC     & 87\% & 44\% & 484 \$ \\ 
					 LECO     & 86\% & 84\% & 273 \$ \\ 
					 LPNT     & 29\% & 67\% & 2487 \$ \\ 
					 LSTR     & 282\% & 73\% & 215 \$ \\ 
					 MAKO     & 16\% & 77\% & 2030 \$ \\ 
					 MANT     & 24\% & 82\% & 507 \$ \\ 
					 MDCO     & 1\% & 84\% & 8878 \$ \\ 
					 MELI     & 15\% & 43\% & 14892 \$ \\ 
					 MIG      & 50\% & 40\% & 413 \$ \\ 
					 MMM      & 53\% & 83\% & -4529 \$ \\ 
					 MOS      & 0\% & 96\% & -10969 \$ \\ 
					 MRTN     & 41\% & 46\% & -846 \$ \\ 
					 MXWL     & 5147\% & 97\% & -1 \$ \\ 
					 NSR      & 11\% & 89\% & 952 \$ \\ 
					 NUS      & 9\% & 54\% & 22755 \$ \\ 
					 NXTM     & 192\% & 58\% & 76 \$ \\ 
					 PBH      & 29\% & 56\% & 419 \$ \\ 
					 PFE      & 5\% & 90\% & 26718 \$ \\ 
					 PG       & 97\% & 101\% & 216 \$ \\ 
					 PNC      & 10\% & 113\% & -6350 \$ \\ 
					 PNY      & 9\% & 75\% & 1483 \$ \\ 
					 PTP      & 41\% & 79\% & 755 \$ \\ 
					 RIGL     & 14\% & 80\% & -2456 \$ \\ 
					 ROC      & 24\% & 68\% & 3124 \$ \\ 
					 ROCK     & 99\% & 56\% & 317 \$ \\ 
					 SF       & 2\% & 95\% & -1029 \$ \\ 
					 SFG      & 23\% & 62\% & 538 \$ \\ 
					 SWN      & 1\% & 104\% & -12890 \$
\end{longtable}
{Table 2-bis. Same as Table 2 with reconstruction of parent trades.}

\begin{longtable}{|l c|l c|l c|}
			stock symbol & prob rejection & stock symbol & prob rejection & stock symbol & prob rejection \\
				\hline
 AA       & 0.96728 &  CPSI     & 0.99755 &  IPAR     & 0.91515 \\ 
 AAPL     & 1 &  CPWR     & 0.99889 &  ISIL     & 0.82535 \\ 
 ADBE     & 0.99999 &  CR       & 0.99897 &  ISRG     & 0.99996 \\ 
 AGN      & 0.99931 &  CRI      & 0.99625 &  JKHY     & 0.99046 \\ 
 AINV     & 0.91669 &  CSCO     & 0.99999 &  KMB      & 0.99999 \\ 
 AMAT     & 0.99682 &  CSE      & 0.99838 &  KR       & 0.99999 \\ 
 AMED     & 0.95456 &  CSL      & 0.98687 &  LANC     & 0.99987 \\ 
 AMGN     & 0.99999 &  CTRN     & 0.99999 &  LECO     & 0.99981 \\ 
 AMZN     & 0.99999 &  CTSH     & 0.99999 &  LPNT     & 0.99999 \\ 
 ANGO     & 0.99997 &  DCOM     & 0.99999 &  LSTR     & 0.99991 \\ 
 APOG     & 0.98397 &  DELL     & 0.78661 &  MAKO     & 0.91648 \\ 
 ARCC     & 0.99519 &  DIS      & 0.99999 &  MANT     & 0.99888 \\ 
 AXP      & 0.99999 &  DK       & 0.99812 &  MDCO     & 0.99998 \\ 
 AYI      & 0.99354 &  DOW      & 0.99999 &  MELI     & 0.90214 \\ 
 AZZ      & 0.99754 &  EBAY     & 1 &  MIG      & 0.98378 \\ 
 BAS      & 0.99999 &  ESRX     & 1 &  MMM      & 0.90883 \\ 
 BHI      & 1 &  EWBC     & 0.99999 &  MOS      & 0.99999 \\ 
 BIIB     & 0.99999 &  FCN      & 0.99989 &  MRTN     & 0.99528 \\ 
 BRCM     & 0.99999 &  FFIC     & 0.99999 &  MXWL     & 0.99998 \\ 
 BRE      & 0.99718 &  FL       & 0.99999 &  NSR      & 0.96399 \\ 
 BXS      & 0.9992 &  FMER     & 0.90467 &  NUS      & 0.9958 \\ 
 BZ       & 0.99976 &  FPO      & 0.99999 &  NXTM     & 0.9997 \\ 
 CB       & 0.99999 &  FRED     & 0.99988 &  PBH      & 0.99717 \\ 
 CBEY     & 0.88292 &  FULT     & 0.99946 &  PFE      & 0.99999 \\ 
 CBT      & 0.98185 &  GAS      & 0.99863 &  PG       & 1 \\ 
 CBZ      & 0.99998 &  GE       & 0.99999 &  PNC      & 0.99999 \\ 
 CDR      & 0.99999 &  GILD     & 1 &  PNY      & 0.99559 \\ 
 CELG     & 0.99999 &  GLW      & 0.89184 &  PTP      & 0.99237 \\ 
 CETV     & 0.9999 &  GOOG     & 1 &  RIGL     & 0.99953 \\ 
 CKH      & 0.99927 &  GPS      & 0.99999 &  ROC      & 0.99412 \\ 
 CMCSA    & 0.99999 &  HON      & 0.99999 &  ROCK     & 0.99036 \\ 
 CNQR     & 0.98594 &  HPQ      & 0.99999 &  SF       & 0.99888 \\ 
 COO      & 0.92067 &  IMGN     & 0.98296 &  SFG      & 0.99996 \\ 
 COST     & 0.99992 &  INTC     & 0.99999 &  SWN      & 1 \\ 
\end{longtable}
{Table 3-bis. Same table as Table 3 with reconstruction of the parent trades.}
}
\end{center}

The discrete time self-financing equation in the general case takes the form
\begin{equation}
\Delta_n X = L_n \Delta_n p \pm c_n\left(\Delta_n L\right) + \Delta_n L \Delta_n p
\end{equation}
where $c_n$ is the transaction cost function associated to the order book. In the case where all the trades happen at the best bid or ask price, this function is equal to $c_n(l) = \frac{s_n}{2}|l|$.

The heuristic we employ to reconstruct parent orders in the ITCH data is as follows. We declare that multiple trades are children of the same parent order if the following three conditions are satisfied
\begin{enumerate}
\item They all have the exact same timestamp.
\item They are all in the same direction (Buy or Sell).
\item In between each trade message, there are no other messages for the same stock. In particular, the order book did not change in between two such consecutive messages.
\end{enumerate}
We believe this to be a rather conservative heuristic. Most trades still happen at the best bid or ask price.

The main effect of the reconstruction of parent orders is to reduce the number of trades where the midprice does not move. This can be seen on the table giving the proportion of trades with, without and with reverse price impact (Table 1-bis). Because of that, our main test is boosted as the instantaneous correlation between liquidity provider inventory and prices is more significant (Table 2-bis). The results of Table 3-bis reflect the higher amount of trades with price impact: the ratio between the adverse selection and the transaction costs increases and sometimes becomes bigger than $1$.

To conclude, reconstructing parent orders complicates the theory, but improves all statistical measures of instantaneous adverse selection.

\section{Proof of the continuous self-financing equation}
The aim of this section is to derive formula \eqref{cont_self_financing} from its discrete version \eqref{self_financing_discrete}. In the process, we shall also derive the continuous-time analog of the instantaneous adverse selection constraint. The key is to let the tick size vanish, assume that the bid-ask spread vanishes with the tick size, and assume that the price and inventory converge to diffusion limits. 

\subsubsection{The spread}
This is conceptually the trickiest part of the argument. If we assume the spread to be of a fixed size, then the transaction cost component of our wealth equation explodes. This is not what we observe empirically in the data however, and we must adapt our set of hypotheses accordingly.

Our claim is that the spread vanishes in order $1/\sqrt{N}$. Our argument for this is to say that the spread is -both empirically and theoretically- a multiple of the price grid and hence the price jumps. Given that in a continuous time model, price jumps are of order $1/\sqrt{N}$ it is consistent to have the spread scale in $1/\sqrt{N}$. We call this hypothesis a 'vanishing bid-ask spread' and denote in the continuous limit by $s_t$ the spread measured in the price-change unit.

\vskip 4pt
The main technical tool we use is the functional law of large number for a discretized process by Jacod and Protter \cite{Jacod}. Let $\phi_{\sigma^2}$ denote the density function of the Gaussian distribution with mean $0$ and variance $\sigma^2$.

\begin{theorem}[(7.2.2) from \cite{Jacod}]
\label{thm_Jacod}
Let $(t, y) \rightarrow F_t(y)$ be an adapted random function that is a.s. continuous in $(t,y)$ and verifies the growth condition $F_t(y) \le C y^2$ for some constant $C$.
Then we have the following convergence u.c.p. as $N\rightarrow \infty$ for any continuous It\^o process $Y$:
\begin{equation*}
\frac{1}{N}\sum_{n=1}^{\lfloor Nt \rfloor} F_{n/N}\left({\sqrt{N}(Y_{(n+1)/N} - Y_{n/N})}\right) \rightarrow \int_0^t \int F_s(y) \phi_{\sigma^2_s}(y)dy \; ds
\end{equation*}
where $\sigma^2_t = \frac{d[Y,Y]_t}{dt}$. 
\end{theorem}
We proceed as follows:
\begin{enumerate}
\item We start from continuous time processes for the inventory $L$, price $p$ and bid-ask spread $s$ as our \emph{data}. 
\item By time discretization, we obtain the \emph{data} to plug into the discrete-time clearing equations listed in section \ref{sec:clearing_condition}, yielding our \emph{discrete time} output relationships.
\item Finally, we take the limit to obtain the diffusion limits of our discrete output to obtain our continuous-time relationships.
\end{enumerate}

In discrete time, the midprice is a pure-jump process, and therefore has finite variations. It is common on larger time scales to consider the price as `zoomed out' enough to be approximated by a diffusion process. Mathematically, this corresponds to a vanishing tick size. Recall that tick size is typically of the order of magnitude of the cent\footnote{Decibasis point for some exchanges in the foreign exchange market.}, that is $10^{-4}$ relative to the typical stock price. Given the relative roughness of the path of inventories when compared to prices it seems reasonable to also expect high-frequency inventories to be modeled by processes with infinite variation.

\subsection{Mathematical Setup}
Let $W$ and $W'$ be two Brownian motions with unspecified correlation structure. We consider two It\^o processes for the price $p$ and the inventory of $L$ a \emph{liquidity provider}:
\begin{equation}
	\begin{cases}
		p_t &= p_0 + \int_0^t \mu_u du + \int_0^t \sigma_u dW_u \\
		L_t &= L_0 + \int_0^t b_u du + \int_0^t l_u dW'_u
	\end{cases}
\end{equation}
where $p_0$ and $L_0$ are given initial conditions and $\mu$, $\sigma$, $b$ and $l$ are adapted continuous processes. Finally, we also assume the existence of an adapted continuous process $s$.

Now consider the discrete approximation $p^N_n = p_{n/N}$ and likewise for $L$, $\mu$, $\sigma$, $b$ and $l$. The interpretation is that $\frac{1}{\sqrt{N}}$ is the tick size, which we formally make vanish. For the bid-ask spread $s$, we define $s^N_n = \frac{1}{\sqrt{N}} s_{n/N}$ in line with our previous comments. Plugging these definitions into our discrete model leads to:

\begin{equation}
\label{fo:3discrete}
	\begin{cases}
		\Delta_n X^N = L^N_n \Delta_n p^N + \frac{s_{n/N}}{2} \frac{1}{\sqrt{N}}|\Delta_n L^N| + \Delta_n p^N \Delta_n L^N \\
		\Delta_n L^N \Delta_n p^N \le 0
	\end{cases}
\end{equation}
where the first equation is understood as the definition of the wealth $X^N$.

\subsection{Main result}
\begin{theorem}
Assuming that relations \eqref{fo:3discrete} hold for every $N\ge 1$, then the limit $\lim_{N\rightarrow \infty} X^N_{\lfloor Nt \rfloor}$ exists for the uniform convergence in probability and defines a process $X_t$ which together with the It\^o processes $p_t$ and $L_t$ satisfy the relationships:
\begin{equation}
\label{fo:3continuous}
	\begin{cases}
		dX_t = L_t dp_t + \frac{s_t l_t}{\sqrt{2 \pi}} dt +d[L,p]_t \\
		d[L,p]_t \le 0
	\end{cases}
\end{equation}
\end{theorem}
\begin{proof}
Using  a localizing sequence of stopping times if needed, we can assume without any loss of generality that the process $s_t$ is bounded by a constant.
The convergence of the discrete approximations of $\int_0^t L_u dp_u$ and $[L,p]_t$ is plain, proving the second relationship.

For the last term of the self-financing equation, we have that
\begin{equation}
\frac{s^N_n}{2} \frac{1}{\sqrt{N}}|\Delta_n L^N| = \frac{1}{2 N} s_{n/N} |\sqrt{N} \Delta_n L^N|
\end{equation}
which allows us to apply Theorem \ref{thm_Jacod} with $F_t(y) = \frac{s_t}{2}|y|$ and $Y_t = L_t$. This proves the self-financing equation.
\end{proof}

\begin{remark}
Technically speaking, nothing prevents us from using the same limiting argument for the hidden part of the order book, simply replacing $p_t$ and $s_t$ by their `hidden' counterparts. Two practical problems appear however. First, measuring the hidden price and spread is difficult. Second, and more importantly, it is unclear by what to replace the price impact inequality, as adverse selection of hidden orders is rarely studied and, in any case, poorly understood.
\end{remark}

\subsection{The case of a liquidity taker}
By symmetry, the corresponding equations for the inventory and wealth of a liquidity taker are 
\begin{equation}
	\begin{cases}
		dX_t = L_t dp_t - \frac{s_t l_t}{\sqrt{2 \pi}} dt +d[L,p]_t \\
		d[L,p]_t \ge 0
	\end{cases}
\end{equation}
Unfortunately, as we already pointed out, both the equations for the liquidity provider and the liquidity taker are only \emph{necessary} conditions. Indeed, unlike with the standard self-financing equation, it is difficult to tell which processes $L$ and $p$ are admissible: we can only derive $X$ once $L$ and $p$ are given.

To give an example of why not all $L$ can be attained, assume the volume on the order book is finite. Then the volatility of $L$ must be bounded by the amount of volume available. Other factors that can come into play to determine which processes $L$ are actually attainable by market participants are: limit order fill rate, instantaneous price recovery and for market orders the ability to predict the next price jump. These factors will directly impact the volatility of $L$ and the possible correlation and quadratic covariation between $L$ and $p$.

Ultimately, supply and demand rule the price $p$ and volume $L$. $X$, however, stems from accounting rules.

\section{Derivation of the Black - Scholes PDE}
In this section we prove the pricing PDE \eqref{pricing_PDE} for a European option in a Black - Scholes model with transaction costs and price impact.

\subsection{Mathematical setup}
Let $W$ be a Brownian motion and $r$ the constant interest rate the cash account is subject to. Assume the midprice $p_t$ to satisfy the stochastic differential equation (SDE)
\begin{equation}
dp_t = \mu p_t dt + \sigma p_t dW_t,
\end{equation}
and let the spread $s_t$ be of the form $s(p_t)$ for some continuous function $s$. 
Given an adapted It\^o process 
\begin{equation}
L_t = L_0 + \int_0^t b_u du + \int_0^t l_u dW_u
\end{equation}
representing an agent's inventory, and a real $K_0$ representing his initial cash endowment, his wealth is defined as
\begin{equation}
X_t = L_0 p_0 + K_0 + \int_0^t L_u dp_u + \int_0^t \left(\sigma p_u - \frac{s(p_u)}{\sqrt{2\pi}}\right)l_u du + \int_0^t r (X_t - p_t L_t)dt
\end{equation}
where we have used equation \eqref{eq:continuous_signed} as $l_t$ is signed. The last term comes from the interest gained on the cash account $K_t = X_t - p_t L_t$.

The objective is to find $X_t$such that we have $X_T = f(p_T)$ at a given terminal time $T$ for a given payoff function $f$.

\subsection{The result}

\begin{theorem}
Let $f$ be a continuous function and $T>0$. Assume $v$ to be the solution of the PDE
\begin{equation}
\frac{\partial v}{\partial t}(t,p) + \frac{\sigma^2 p^2}{2}\left(\frac{s(p)\sqrt{2}}{\sigma p\sqrt{\pi}} - 1\right)\frac{\partial^2 v}{\partial p^2}(t,p) + r p \frac{\partial v}{\partial p}(t,p) = r v(t,p) 
\end{equation}
with terminal condition $v(T,p) = f(p)$. Then  $L_t = \frac{\partial v}{\partial p}(t,p_t)$, and the choice $K_0 = v(0,p_0) - \frac{\partial v}{\partial p}(0,p_0) p_0$ leads to $X_T = f(p_T)$. Furthermore, 
\begin{equation}
l_t = \sigma p_t \frac{\partial^2 v}{\partial p^2}(t,p_t).
\end{equation}
\end{theorem}
\begin{proof}
Consider the process $\tilde{X}_t = e^{-r(t-T)}X_t$. Choosing the ``delta'' for the holdings $L_t$, namely 
\begin{equation}
L_t = \frac{\partial v}{\partial p}(t,p_t)
\end{equation}

leads to
\begin{equation}
l_t = \sigma p_t \frac{\partial^2 v}{\partial p^2}(t,p_t)
\end{equation}
and hence
\begin{align*}
d\tilde{X}_t &= e^{-r(t-T)}\left((\mu-r) p_t \frac{\partial v}{\partial p}(t,p_t) - \sigma^2 p^2_t \left(\frac{s(p_t)}{\sigma p_t\sqrt{2\pi}}  -1\right) \frac{\partial^2 v}{\partial p^2}(t,p_t)  \right)dt \\
						 &\quad+ e^{-r(t-T)}\frac{\partial v}{\partial p}(t,p_t) p_t \sigma dW_t \\
						 &= e^{-r(t-T)}\left(- r v(t,p_t) + \mu p_t \frac{\partial v}{\partial p}(t,p_t) + \frac{1}{2}\sigma^2 p^2_t\frac{\partial^2 v}{\partial p^2}(t,p_t) + \sigma p_t\frac{\partial v}{\partial t}(t,p_t)\right)dt  \\
						 &\quad+ e^{-r(t-T)}\frac{\partial v}{\partial p}(t,p_t) p_t \sigma dW_t \\
						 &= d(e^{-r(t-T)} v(t,p_t))
\end{align*}
As the initial values match, we have that $X_t = v(t,p_t)$. This concludes.
\end{proof}

\end{document}